%% file: arXiv-1303.1646v2.tex
\documentclass[11pt,a4paper]{article}
\usepackage{a4wide}
\usepackage{latexsym}
\usepackage{amsmath}
\usepackage{amssymb}
\usepackage{verbatim}
\usepackage{paralist}
\usepackage{multirow}
\usepackage{multicol}
\usepackage{amsthm}
\usepackage{url}
\usepackage[normalem]{ulem}

\newtheorem{theorem}{Theorem}
\newtheorem{lemma}{Lemma}
\newtheorem{proposition}{Proposition}
\newtheorem{definition}{Definition}
\newtheorem*{definition*}{Definition}
\newtheorem{corollary}{Corollary}
\newtheorem{remark}{Remark}

\renewcommand{\vec}[1]{\mathbf{#1}}

\newcommand{\dpa}{Discriminatory Auction}
\newcommand{\upa}{Uniform Price Auction}
\newcommand{\bpoa}{Bayesian Price of Anarchy}

\newcommand{\xos}{{\bf XOS}}

\usepackage{color}
\newcommand{\otrem}[1]{#1}
\newcommand{\vorem}[1]{#1}

\newenvironment{proofof}[1]
      {\medskip\noindent{\em Proof of #1.}\ }
      {\hfill$\Box$\medskip}

\allowdisplaybreaks

\begin{document}
\bibliographystyle{plain}
\date{}
\title{
On the Inefficiency of Standard Multi-Unit Auctions
\protect\footnote{
A preliminary version of this work appeared in~\cite{deKeijzer13}.
}
}
\author{%
Bart de Keijzer$^1$\qquad
Evangelos Markakis$^3$\qquad
Guido Sch\"afer$^{2}$\qquad
Orestis Telelis$^4$
\protect\footnote{
Research partially developed while the author was affiliated to Athens
University of Economics and Business, through his work on the project ``DDCOD''
(PE6-213). The project was implemented within the framework of the Action
``Supporting Postdoctoral Researchers'' of the Operational Program ``Education
and Lifelong Learning'' (Action's Beneficiary: General Secretariat for Research
and Technology), and was co-financed by the European Union (European Social Fund
-- ESF) and the Greek State.
}
\medskip\\
$^1$Department of Computer, Control and Management Engineering,\\
Sapienza University of Rome, Italy\\
{\tt dekeijzer@dis.uniroma1.it}\smallskip\\
$^2$CWI and VU University Amsterdam, The Netherlands\\
{\tt g.schaefer@cwi.nl}\smallskip\\
$^3$Dept. of Informatics, Athens University of Economics and Business, Greece\\
{\tt markakis@gmail.com}\smallskip\\
$^4$Department of Digital Systems, University of Piraeus, Greece\\
{\tt telelis@gmail.com}
}

\maketitle

\begin{abstract}
We study two standard multi-unit auction formats for allocating multiple units
of a single good to multi-demand bidders. The first one is the Discriminatory
Auction, which charges every winner his winning bids. The second is the Uniform
Price Auction, which determines a uniform price to be paid per unit. Variants of
both formats find applications ranging from the allocation of state bonds to
investors, to online sales over the internet, facilitated by popular online
brokers. 

For these formats, we consider two bidding interfaces: (i) standard bidding,
which is most prevalent in the scientific literature, and (ii) uniform bidding,
which is more popular in practice. In this work, we evaluate the economic
inefficiency of both multi-unit auction formats for both bidding interfaces, by
means of upper and lower bounds on the Price of Anarchy for pure Nash equilibria
and mixed Bayes-Nash equilibria. Our developments improve significantly upon
bounds that have been obtained recently in [Markakis, Telelis, ToCS 2014] and
[Syrgkanis, Tardos, STOC 2013] for submodular valuation functions. Moreover, we
consider for the first time bidders with subadditive valuation functions for
these auction formats. Our results signify that these auctions are nearly
efficient, which provides further justification for their use in practice.
\end{abstract}

\input{introduction.tex}

\input{definitions.tex}

\input{pne.tex}

\input{bne.tex}

\input{smoothness.tex}

\input{conclusions.tex}

\bibliography{biblio}
\end{document}

%% file: introduction.tex
\section{Introduction}
\label{section:introduction}
We study standard multi-unit auction formats for allocating multiple units of a
single good to multi-demand bidders. Multi-unit auctions are one of the most
widespread and popular tools for selling identical units of a good with a single
auction process. In practice, they have been in use for a long time, one of
their most prominent applications being the auctions offered by the U.S. and
U.K. Treasuries for selling bonds to investors, see e.g., the U.S. treasury
report \cite{treasury}. In more recent years, they are also implemented by
various online brokers~\cite{ebid,Ockenfels06}. In the literature, multi-unit
auctions have been a subject of study ever since the seminal work of
Vickrey~\cite{Vickrey61} (although the need for such a market enabler was
conceived even earlier, by Friedman, in~\cite{Friedman60}) and the success of
these formats has led to a resurgence of interest in auction design. 

There are three simple {\em Standard Multi-Unit Auction} formats that have
prevailed and are being implemented; these are the {\em Discriminatory Auction},
the {\em Uniform Price Auction} and the {\em Vickrey Multi-Unit Auction}. All
three formats share a common allocation rule and bidding interface and have seen
extensive study in auction theory~\cite{Krishna02,Milgrom04}. Each bidder under
these formats is asked to issue a sequence of non-increasing marginal bids, one
for each additional unit \otrem{(i.e., a submodular curve)}. For an auction of
$k$ units, the $k$ highest marginal bids win, and each grants its issuing bidder
a single unit. The  formats differ in the way that payments are determined for
the winning bidders. The Discriminatory Auction prescribes that each bidder pays
the sum of his winning bids. The Uniform Price Auction charges the lowest
winning or highest losing marginal bid per allocated unit. The Vickrey auction
charges according to an instance of the Clarke payment rule (thus being a
generalization of the well known single-item Second-Price auction).

Except for the Vickrey auction, which is truthful and efficient, the other two
formats suffer from a {\em demand reduction} effect~\cite{Ausubel02}, whereby
bidders may have incentives to understate their value, so as to receive less
units at a better price. This effect is amplified when bidders have
non-submodular valuation functions, since the bidding interface forces them to
encode their value within a submodular bid vector. Even worse, in many practical
occasions bidders are asked for a {\em uniform bid} per unit together with an
upper bound on the number of desired units. In such a setting, each bidder is
required to ``compress'' his valuation function into a bid that scales linearly
with the number of units. The mentioned allocation and pricing rules apply also
in this {\em uniform bidding} setting, thus yielding different versions of
Discriminatory and \upa s. Despite the volume of research from the economics
community~\cite{Ausubel02,Noussair95,Engelbrecht-Wiggans98,Binmore00,Reny99,Bresky08}
and the widespread popularity of these auction formats, the first attempts of
quantifying their economic efficiency are only very
recent~\cite{Markakis12,ST13,Markakis14}. There has also been no study of these
auction formats for non-submodular valuations, as noted by Milgrom
\cite{Milgrom04}. 

\begin{table}[t]
\begin{center}
\begin{tabular}{c@{\quad}|@{\quad}c@{\quad}c}
\multirow{3}{*}{{\bf Valuation Functions}} & \multicolumn{2}{c}{{\bf Auction Format}} \\
& \multicolumn{2}{c}{$(${\tt bidding: standard $|$ uniform}$)$}\\
 & {\em Discriminatory Auction} & {\em Uniform Price Auction}  \\
\hline
\multirow{3}{*}{{\em Submodular}} & & \\
 & $\displaystyle\frac{e}{(e-1)}$ & $3.1462$  \\
 & \\
\multirow{3}{*}{{\em Subadditive}} & & \\
 & $2$ $\Bigl|\Bigr.$ $\displaystyle\frac{2e}{(e-1)}$ & $4$ $\Bigl|\Bigr.$ $6.2924$\\
 &  
\end{tabular}
\end{center}
\caption{
Upper bounds on the (Bayesian) economic inefficiency of multi-unit auctions.}
\label{table:upper-bounds}
\end{table}

\subsection{Contribution}
\label{subsection:contribution}

We study the inefficiency of the \dpa\ and \upa\ under the standard and uniform
bidding interfaces. Our main results are improved inefficiency bounds for
bidders with submodular valuation functions and new bounds for bidders with
subadditive valuation functions. To the best of our knowledge, for subadditive
valuation functions our bounds provide the first quantification of the
inefficiency of these auction formats. The results are summarized in
Table~\ref{table:upper-bounds}. Our bounds indicate that these auctions are
nearly efficient, which paired with their simplicity provides further
justification for their use in practice. 

%Our focus is mostly on \otrem{quantifying} the efficiency of Bayes-Nash
%equilibria (see Section~\ref{section:sm}); \otrem{however, 
\otrem{As a warmup, we begin by discussing pure Nash equilibria, in
Section~\ref{section:pne}. In particular, we investigate their properties and
establish inefficiency bounds for both auction formats and both bidding
interfaces. An interesting consequence is that discriminatory pricing always
leads to efficient pure Nash equilibria.}

%, but we also discuss pure Nash equilibria (see {\bf Appendix A}).
Subsequently, we proceed to develop our main results, for the inefficiency of
Bayes-Nash equilibria, in Section~\ref{section:sm}. For submodular valuation
functions, we derive upper bounds of $\frac{e}{e-1}$ and $3.1462 <
\frac{2e}{e-1}$ for the Discriminatory and the Uniform Price Auctions,
respectively. These improve upon the previously best known bounds of
$\frac{2e}{e-1}$~\cite{ST13} and $4$~\cite{Markakis14}. \otrem{Regarding lower
bounds, for the Uniform Price Auction, we establish a lower bound of $2$, which
holds even for pure Nash equilibria. We also prove a lower bound of
$\frac{e}{e-1}$ for the \dpa, with respect to the currently known proof
techniques~\cite{ST13,Feldman12,Christodoulou08,Bhawalkar11,Hassidim11}. As a
consequence, unless the upper bound of $\frac{e}{e-1}$ for the \dpa\ is tight,
its improvement requires the development of novel tools.}

We then move to subadditive valuation functions, where we obtain bounds of
$\frac{2e}{e-1}$ and $6.2924<\frac{4e}{e-1}$ for Discriminatory and Uniform
Price Auctions respectively, independently of the bidding interface. These, we
can improve by conditioning to the standard bidding interface, to $2$ and $4$,
respectively, by adapting a recent technique from~\cite{Feldman12}.  In Section
\ref{sec:z-smoothness} we discuss further applications of our results in
connection with the smoothness framework of~\cite{ST13}. In particular, some of
our bounds carry over to simultaneous and sequential compositions of such
auctions (see Table \ref{table:upper-bounds2} in Section
\ref{sec:z-smoothness}).

\section{Related Work}

A first quantification of the Uniform Price Auction's inefficienct for bidders
with submodular valuation functions was presented in~\cite{Markakis14}. In
particular, it was shown that, under the standard bidding interface, this
auction format admits pure Nash equilibria in undominated strategies, with at
most $\frac{e}{e-1}$ times less welfare than the socially optimal allocation.
This bound was shown to be tight. For the inefficiency of (mixed) Bayes-Nash
equilibria an upper bound of $4$ was developed. The first quantification of the
Discriminatory Auction's inefficiency for submodular bidders was developed by
Syrgkanis and Tardos in~\cite{ST13}. The authors presented several extensions of
Roughgarden's {\em smoothness technique}~\cite{Roughgarden09}, which they used
for proving inefficiency upper bounds for several auction formats and
compositions of them.  For the Discriminatory Auction with submodular bidders in
particular they proved an upper bound of $\frac{2e}{e-1}$. More developments
from~\cite{ST13} we discuss below. Recently, Christodoulou {\em et al.}
constructed in~\cite{Christodoulou13} a lower bounding example of $1.109$ for
the Price of Anarchy of mixed Nash equilibria of the Discriminatory Auction, in
the full information settting, when bidders have submodular valuation functions. 

Let us note that, concerning bidders with {\em subadditive} valuation
functions, there is no previous work on quantifying the inefficiency of
multi-unit auction formats. Since the appearance of the preliminary version of
our work in~\cite{deKeijzer13}, Christodoulou {\em et al.} showed
in~\cite{Christodoulou13} a lower bound of $2$ for the Price of Anarchy of mixed
Nash equilibria of the Discriminatory Auction for bidders with subadditive
valuation functions, which matches our upper bound of $2$. Finally, among the
studies concerning the inefficiency of multi-unit auctions, we mention the
extensive and definitive treatment of the {\em Generalized Second-Price} Auction 
format by Caragiannis {\em et al.}~\cite{Caragiannis14}, wherein almost tight
bounds are derived (in the full and incomplete information settings) by
means of {\em smoothness} techniques.

The multi-unit auction formats that we examine here present technical and
conceptual resemblance to the {\em Simultaneous Auctions} format that has
received significant attention
recently~\cite{Feldman12,Christodoulou08,Bhawalkar11,Hassidim11,ST13,Christodoulou13,Roughgarden14}.
However, upper bounds in this setting do not  carry over to our format.
Simultaneous auctions were first studied by Christodoulou, Kovacs and
Schapira~\cite{Christodoulou08}. The authors proposed that each of a collection
of distinct goods, with one unit available for each of them, is sold in a
distinct {\em Second Price Auction}, simultaneously and independently of the
other goods. Bidders in this setting may have combinatorial valuation functions
over the subsets of goods, but they are forced to bid separately for each good.
For bidders with fractionally subadditive valuation functions, they proved a
tight upper bound of $2$ on the mixed Bayesian Price of Anarchy of the
Simultaneous Second Price Auction. Bhawalkar and Roughgarden~\cite{Bhawalkar11}
extended the study of inefficiency for subadditive bidders and showed an upper
bound of $O(\log m)$ which was recently reduced to $4$ by Feldman {\em et
al.}~\cite{Feldman12}. For arbitrary valuation functions, Fu, Kleinberg and
Lavi~\cite{Fu12} proved an upper bound of $2$ on the inefficiency of pure Nash
equilibria, when they exist.

Hassidim {\em et al.}~\cite{Hassidim11} studied {\em Simultaneous First Price
Auctions}. They showed that pure Nash equilibria in this format are always
efficient, when they exist. They proved constant upper bounds on the
inefficiency of mixed Nash equilibria for (fractionally) subadditive valuation
functions and $O(\log m)$ and $O(m)$ for the inefficiency of mixed Bayes-Nash
equilibria for subadditive and arbitrary valuation functions. Syrgkanis showed
in~\cite{Syrgkanis12a} that this format has Bayesian Price of Anarchy
$\frac{e}{e-1}$ for fractionally subadditive valuation functions. Feldman {\em
et al.}~\cite{Feldman12} proved an upper bound of $2$ for subadditive ones. 

Recently, Syrgkanis and Tardos~\cite{ST13} and Roughgarden~\cite{Roughgarden12}
independently developed extensions of the {\em smoothness technique} for games
of incomplete information. In \cite{ST13}, these ideas are further developed for
analyzing the inefficiency of simultaneous and sequential {\em compositions} of
simple auction mechanisms. They demonstrate extensive applications of their
techniques on welfare analysis of standard multi-unit auction formats {\em and}
their (simultaneous and sequential) compositions. For submodular valuation
functions, they prove inefficiency upper bounds of $\frac{2e}{e-1}$ and
$\frac{4e}{e-1}$ for the \dpa\ and \upa, respectively. Here, we improve upon
these results; our improvements carry over to simultaneous and sequential
compositions as well.

%% file: definitions.tex
\section{Definitions and Preliminaries}
\label{section:definitions}

We consider auctioning $k$ units of a single good to a set $[n] = \{1,...,n\}$
of $n$ bidders, indexed by $i=1,\dots, n$. Every bidder $i\in[n]$ has a
non-negative non-decreasing private valuation function $v_i:(\{0\} \cup
[k])\mapsto\Re^+$ over quantities of units, where $v_i(0) = 0$. We denote by
$\vec{v}=(v_1,\dots,v_n)$ the {\em valuation function profile} of bidders. We
consider in particular (symmetric) \emph{submodular} and \emph{subadditive}
functions:
 
\begin{definition}
A valuation function $f:(\{0\}\cup[k])\mapsto\Re^+$ is called:
\begin{compactitem}
\item submodular iff for every $x<y$, $f(x)-f(x-1)\geq f(y)-f(y-1)$.
\item subadditive iff for every $x,y$, $f(x+y)\leq f(x)+f(y)$.
\end{compactitem}
\end{definition}

\noindent The class of submodular functions is strictly contained in the class
of subadditive ones~\cite{Lehmann06}. For any non-negative non-decreasing
function $f:(\{0\}\cup[k])\mapsto\Re^+$ and any integers $x, y\in [k], x < y$,
the following are known to hold:
If $f$ is submodular, then $f(x)/x \geq f(y)/y$. 
If $f$ is subadditive, then $f(x)/x\geq f(y)/(x+y)$.

A valuation function $v_i$ can be specified by a vector
$\vec{m}_i=(m_i(1),...,m_i(k))$ of the {\em marginal values} $m_i(j) = v_i(j) -
v_i(j-1)$ of bidder $i$, for each additional unit in his allocation (if $v_i$ is
submodular, $m_i(j)\geq m_i(j+1)$). 

\subsection*{Standard Multi-Unit Auctions}

The {\em standard format}, as described in auction
theory~\cite{Krishna02,Milgrom04}, prescribes that each bidder $i\in[n]$ submits
a vector of $k$ non-negative non-increasing  \emph{marginal bids}
$\vec{b}_i=(b_i(1),\dots, b_i(k))$ with $b_i(1) \geq \dots \geq b_i(k)$. We will
often refer to these simply as \emph{bids}. In the {\em uniform bidding format},
each bidder $i$ submits only a single bid $\bar{b}_i$ along with a quantity $q_i
\le k$, the interpretation being that $i$ is willing to pay at most $\bar{b}_i$
per unit for up to $q_i$ units. 

The allocation rule of standard multi-unit auctions grants the issuer of each of
the $k$ highest (marginal) bids a distinct unit per winning bid. The pricing
rule differentiates the formats. Let $x_i(\vec{b})$ be the number of units won
by bidder $i$ under profile $\vec{b}=(\vec{b}_1,\dots,\vec{b}_n)$. We study the
following two pricing rules:

\bigskip

%%%%%%%%%%%%%%%%%%%%%%%%%%%%%%%%%%%%%%%%%%%%%%%%%%%%%%%%%%%%%%%%%%%%%%%%%%
%%%%%%%%%%%%%%%%%%%%%%%%%%%%%%%%%%%%%%%%%%%%%%%%%%%%%%%%%%%%%%%%%%%%%%%%%%
%%%%%%%%%%%%%%%%%%%%%%%%%%%%%%%%%%%%%%%%%%%%%%%%%%%%%%%%%%%%%%%%%%%%%%%%%%
\begin{comment}
\noindent
\begin{tabular}{l@{\quad}r}
\begin{minipage}[t]{0.47\textwidth}
\noindent {\em\textbf{(i)} Discriminatory Pricing.} 
Every bidder $i$ pays for every unit a price equal to his corresponding winning bid, i.e., the utility of $i$ is \\
$\mbox{}\quad
\textstyle u_i^{v_i}(\vec{b}) = v_i\left(x_i(\vec{b})\right)-\sum_{j=1}^{x_i(\vec{b})}b_i(j).
$
\end{minipage}
&
\begin{minipage}[t]{0.47\textwidth}
\noindent {\em\textbf{(ii)} Uniform Pricing.} \qquad
Every bid\-der $i$ pays for every unit a price equal to the \emph{highest losing bid} $p(\vec{b})$, i.e., the utility of $i$ is \\
$\mbox{}\quad
u_i^{v_i}(\vec{b}) = v_i\left(x_i(\vec{b})\right) - x_i(\vec{b}) p(\vec{b}).
$
\end{minipage}
\end{tabular}
\end{comment}
%%%%%%%%%%%%%%%%%%%%%%%%%%%%%%%%%%%%%%%%%%%%%%%%%%%%%%%%%%%%%%%%%%%%%%%%%%
%%%%%%%%%%%%%%%%%%%%%%%%%%%%%%%%%%%%%%%%%%%%%%%%%%%%%%%%%%%%%%%%%%%%%%%%%%
%%%%%%%%%%%%%%%%%%%%%%%%%%%%%%%%%%%%%%%%%%%%%%%%%%%%%%%%%%%%%%%%%%%%%%%%%%

\noindent {\em\textbf{(i)} Discriminatory Pricing.} Every bidder $i$ pays for
every unit a price equal to his corresponding winning bid, i.e., the utility of
$i$ is:
%%%
$
u_i^{v_i}(\vec{b}) = v_i\left(x_i(\vec{b})\right)-\sum_{j=1}^{x_i(\vec{b})}b_i(j).
$

\bigskip

\noindent {\em\textbf{(ii)} Uniform Pricing.} Every bid\-der $i$ pays for every
unit a price equal to the \emph{highest losing bid} $p(\vec{b})$, i.e., the
utility of $i$ is: 
%%%
$
u_i^{v_i}(\vec{b}) = v_i\left(x_i(\vec{b})\right) - x_i(\vec{b}) p(\vec{b}).
$

\bigskip

\noindent For a bidding profile $\vec{b}$, the produced allocation
$\vec{x}(\vec{b})=(x_1(\vec{b}), x_2(\vec{b}),\dots,x_n(\vec{b}))$ has a {\em
social welfare} equal to the bidders' total value:
%%%
$
SW(\vec{v}, \vec{b}) = \sum_{i=1}^nv_i(x_i(\vec{b}))
$.
%%%
The (pure) Price of Anarchy is the worst case ratio, over all pure Nash
equilibrium profiles $\vec{b}$, of the optimal social welfare over $SW(\vec{v},
\vec{b})$. 

\subsection*{Incomplete Information}

Under the incomplete information model of Harsanyi, the valuation function
$\vec{v}_i$ of bidder $i$ is drawn from a finite set $V_i$ according to a
discrete probability distribution $\pi_i: V_i\rightarrow [0,1]$ (independently
of the other bidders); we will write $\vec{v}_i \sim\pi_i$. The actual drawn
valuation function of every bidder is {\em private}. A valuation profile
$\vec{v}=(\vec{v}_1,\ldots, \vec{v}_n) \in {\cal V} = \times_{i \in [n]} V_i$ is
drawn from a {\em publicly known distribution} $\pi:{\cal V} \rightarrow [0,1]$,
where $\pi$ is the product distribution of $\pi_1, \ldots, \pi_n$, i.e.,
$\pi(\vec{v}) \mapsto \prod_{i \in [n]} \pi_i(\vec{v}_i)$. Every bidder $i$
knows his own valuation function but does not know the valuation function
$\vec{v}_{i'}$ drawn by any other bidder $i'\neq i$. Bidder $i$ may only use his
knowledge of $\pi$ to estimate $\vec{v}_{-i}$. Given the publicly known
distribution $\pi$, the (possibly mixed) strategy of every bidder is a function
of his own valuation $\vec{v}_i$, denoted by $B_i(\vec{v}_i)$. $B_i$ maps a
valuation function $\vec{v}_i \in V_i$ to a {\em distribution}
$B_i(\vec{v}_i)=B_i^{v_i}$, over all possible bid vectors for $i$. In this case
we will write $\vec{b}_i\sim B_i^{v_i}$, for any particular bid vector
$\vec{b}_i$ drawn from this distribution. We also use the notation
$\vec{B}_{-i}^{\vec{v}_{-i}}$, to refer to the vector of randomized strategies
of bidders other than $i$, under profile $\vec{v}_{-i}$. 

A {\it Bayes-Nash equilibrium} (BNE) is a strategy profile $\vec{B}=(B_1,\dots,
B_n)$ such that for every bidder $i$ and for every valuation $\vec{v}_i$,
$B_i(\vec{v}_i)$ maximizes the utility of $i$ in expectation, over the
distribution of the other bidders' valuations $\vec{w}_{-i}$ {\em given
$\vec{v}_i$} and over the distribution of $i$'s own and the other bidders'
strategies, $\vec{B}^{(\vec{v}_i,\vec{w}_{-i})}$, i.e., for every pure strategy
$\vec{c}_i$ of $i$:
%%%
\[
\mathbb{E}_{\substack{
\vec{w}_{-i}|\vec{v}_i,\\
\vec{b}\sim\vec{B}^{(\vec{v}_i,\vec{w}_{-i})}
}}
\Bigl[u_i^{v_i}(\vec{b})\Bigr]
\geq 
\mathbb{E}_{\substack{
\vec{w}_{-i}|\vec{v}_i,\\
\vec{b}_{-i}\sim\vec{B}^{\vec{w}_{-i}}
}}
\Bigl[u_i^{v_i}(\vec{c}_i,\vec{b}_{-i})\Bigr]
\]
%%%
\noindent where $\mathbb{E}_{\vec{v}}$ and $\mathbb{E}_{\vec{w}_{-i}|\vec{v}_i}$
denote the expectation over the distributions $\pi$ and $\pi(\cdot|\vec{v}_i)$
(i.e., given $\vec{v}_i$), respectively. 

Fix a valuation profile $\vec{v}\in{\cal V}$ and consider a (mixed) bidding
configuration $\vec{B}^{\vec{v}}$ under $\vec{v}$. The Social Welfare
$SW(\vec{v}, \vec{B}^{\vec{v}})$ under $\vec{B}^{\vec{v}}$ when the valuations
are $\vec{v}$ is defined as the expectation over the bidding profiles chosen by
the bidders from their randomized strategies, i.e., 
%%%
$
SW(\vec{v}, \vec{B}^{\vec{v}})
=
\mathbb{E}_{\vec{b}\sim\vec{B}^{\vec{v}}}
\left[\sum_iv_i(x_i(\vec{b}))\right]
$. 
%%%
The {\em expected} Social Welfare in Bayes-Nash equilibrium $\vec{B}^{\vec{v}}$
is then $\mathbb{E}_{\vec{v}\sim\pi}[SW(\vec{v},\vec{B}^{\vec{v}})]$. The
socially optimum assignment under valuation profile $\vec{v}\in{\cal V}$ will be
denoted by $\vec{x}^{\vec{v}}$. The {\em expected} optimum social welfare is
then $\mathbb{E}_{\vec{v}}[SW(\vec{v},\vec{x}^{\vec{v}})]$. Under these
definitions, we will study the {\em Bayesian Price of Anarchy}, i.e., the worst
case ratio
%%%
$\mathbb{E}_{\vec{v}}[SW(\vec{v}, \vec{x}^{\vec{v}})] /
\mathbb{E}_{\vec{v}}[SW(\vec{v}, \vec{B}^{\vec{v}})]$ 
%%%
over all possible product distributions $\pi$ and Bayes-Nash equilibria
$\vec{B}$ for $\pi$.
%\footnote{As in previous works~\cite{Christodoulou08,Feldman12}, we ensure
%existence of Bayes-Nash equilibria in our auction formats by assuming that
%bidders have bounded and finite strategy spaces, e.g., derived through
%discretization. Our bounds on the auctions' Bayesian inefficiency hold for
%sufficiently fine discretizations (see also Appendix D of~\cite{Feldman12}).}

\vorem{
Following the practice in previous related works, when analyzing the Uniform
Price Auction we make the assumption of \emph{no-overbidding}, i.e., that no
bidder bids more than his value for any number of units in {\em any profile};
formally, for every $s \in [k]$, $\sum_{j \in [s]} b_i(j) \leq v_i(s)$. This
assumption is justified by the fact that overbidding strategies are generally
weakly dominated, in that they may cause violation of individual rationality. On
the other hand, although the Uniform Price Auction has individually rational
equilibrium profiles wherein overbidding occurs, these are generally unboundedly
inefficient (much like in the well-known standard example for the single-item
Vickrey auction). In our analysis, we will use $\beta_j(\vec{b})$ to refer to
the $j$-th lowest winning bid under profile $\vec{b}$; thus
$\beta_1(\vec{b})\leq \dots \le \beta_k(\vec{b})$.  }
%%%
\vorem{
\begin{remark}
\quad\newline
\begin{itemize}
\item[(I)] As in previous related work~\cite{Christodoulou08,Feldman12}, we can
ensure existence of Bayes-Nash equilibria in the auction formats that we study,
if we assume that bidders have bounded and finite strategy spaces, that emerge
e.g. after some discretization. To this end, our bounds on the auctions'
Bayesian Price of Anarchy concern essentially such sufficiently fine discrete
approximations, of our setting with continuous strategy spaces. For a relevant
detailed discussion we refer the reader to~\cite{Feldman12} (Appendix D). 
%%%
\item[(II)] 
%As discussed above, in order to obtain meaningful inefficiency bounds for the
%Uniform Price Auction, we need to use the {\em no-\-overbidding assumption},
%i.e., to assume that bidders never outbid their value for any number of units
%that they bid on. 
For both auction formats, the Discriminatory and the Uniform Price, we analyze
equilibrium profiles wherein no bidder outbids his value for the number of units
he receives in equilibrium; however, in a Discriminatory Auction, we do not need
to restrict the bidders' strategy spaces with no-overbidding, because the
auction does not have equilibria wherein overbidding occurs (contrary to the
Uniform Price Auction).
\end{itemize}
\end{remark}
}

%% file: pne.tex
%\section*{APPENDIX A: Pure Nash Equilibria}
\section{Pure Nash Equilibria}
\label{section:pne}

In this section we discuss the properties of pure Nash equilibria of the two
multi-unit auction formats, under both the standard and uniform bidding
interfaces. As we show, pure Nash equilibria are always efficient under the
Discriminatory Auction, unlike the Uniform Price Auction.

\subsection{Uniform Pricing}

Pure Nash equilibria of the Uniform Price Auction have been analyzed recently
in~\cite{Markakis12}. It is well known that a socially optimum allocation can
always be implemented as a pure Nash equilibrium of this auction. Moreoever, by
its loose association to the Vickrey Second-Price Auction, the Uniform Price
Auction retains certain properties in {\em undominated strategies}, for bidders
with submodular valuation functions. In particular, for any $j\in[k]$ , it is
always a {\em weakly dominated strategy} for any bidder $i$ to issue a marginal
bid $b_i(j)> m_i(j)$. Moreoever, issuing $b_i(1)\neq v_i(1)$ is also weakly
dominated. Markakis and Telelis~\cite{Markakis12} showed that pure Nash
equilibria of the Uniform Price Auction in undominated strategies have
inefficiency $\frac{e}{e-1}$. 

In this paper we consider more general no-overbidding strategies, inclusively of
the uniform bidding ones. The following simple example shows that the
inefficiency of pure Nash equilibria of the Uniform Price Auction in general
no-overbidding strategies is slightly higher than the bound of~\cite{Markakis12}
(under the standard or uniform bidding interface):

\begin{theorem}
The Price of Anarchy is at least $\frac{2k-1}{k}$ for the \upa\ with submodular
valuation functions.
\label{theorem:upa-sm-lb}
\end{theorem}

\begin{proof}
Consider $2$ bidders with submodular valuation functions, as follows:

\[
v_1(x)=\left\{
\begin{array}{ll}
k, & \mbox{if $x\geq 1$}\\
0, & \mbox{otherwise}
\end{array}
\right.
\mbox{\quad and\quad}
v_2(x) = x
\]

\noindent The socially optimal allocation achieves welfare $2k-1$, by granting
one unit to bidder $1$ and $k-1$ units to bidder $2$. For a pure Nash
equilibrium consider a uniform bid $\langle b_1=1,k\rangle$ of bidder $1$,
claiming at most $k$ units, and $b_2=0$ for bidder $2$. Clearly bidder $1$
obtains maximum utility, equal to $k$, in this configuration; for every
additional unit that bidder $2$ could obtain by deviating appropriately, he
would gain additional value equal to $1$ and increase his payment by exactly
$1$. Thus, bidder $2$ also cannot improve his utility. The social welfare in
this pure Nash equilibrium configuration is $k$. This implies a Price of Anarchy
that approaches $2$ with the growth of $k$.
\end{proof}

For both bidding interfaces, upper bounds for the inefficiency of pure Nash
equilibria will, emerge from our results for the -- more general -- incomplete
information model (Section~\ref{section:sm}), in combination with the following
lemma:

\begin{lemma}
%For every non-overbidding pure Nash equilibrium $\vec{b}$ of the Uniform Price
%Auction with submodular bidders, there exists a pure Nash equilibrium
%$\vec{\bar{b}}$ of the auction with uniform bidding, such that
%$SW(\vec{\bar{b}})=SW(\vec{b})$. Moreover, 
For arbitrary valuation functions, every no-overbidding pure Nash equilibrium of
the auction with uniform bidding is also a pure Nash equilibrium of the auction
under standard bidding, subject to no-overbidding deviations.
\label{lemma:ub-sb} 
\end{lemma}

\begin{comment}
\begin{lemma}
For every pure Nash equilibrium $\vec{b}$ in undominated strategies, of the
Uniform Price Auction with submodular bidders, there exists a pure Nash
equilibrium $\vec{\bar{b}}$ of the auction with uniform bidding, such that
$SW(\vec{\bar{b}})=SW(\vec{b})$. Moreover, for arbitrary valuation functions,
every no-overbidding pure Nash equilibrium of the auction with uniform bidding
is also a pure Nash equilibrium of the auction under standard bidding, subject
to no-overbidding deviations.  \label{lemma:ub-sb} 
\end{lemma}
\end{comment}

\begin{proof}
%We may assume that $\vec{b}$ has the form prescribed by Proposition 2
%in~\cite{Markakis12}; i.e., for every winning bidder $i$, $b_i(j)=m_i(j)$, if
%$j\leq x_i$ and $b_i(j)=0$ for $j>x_i$. For every losing bidder,
%$\vec{b}_i=(m_i(1),0,\dots,0)$. Then, the uniform price $p(\vec{b})$ is equal
%to $m_i(1)=v_i(1)$ for some $i$, or $0$. We convert $\vec{b}$ into a uniform
%bidding profile $\vec{\bar{b}}\equiv\langle \bar{b}_i,q_i\rangle$ as follows:
%set $\bar{b}_i=v_i(x_i(\vec{b}))/x_i(\vec{b})$ and $q_i=x_i(\vec{b})$ if
%$x_i(\vec{b}) \geq 1$;  otherwise set $\bar{b}_i = m_i(1)$ and $q_i = 1$. Then,
%observe that $p(\vec{\bar{b}})=p(\vec{b})$. Because for every {\em winning}
%bidder $i$ under $\vec{b}$, we have $b_i(x_i(\vec{b}))=m_i(x_i(\vec{b}))\geq
%p(\vec{b})$, it is also $\bar{b}_i\geq p(\vec{b})=p(\vec{\bar{b}})$. Thus,
%$x_i(\vec{\bar{b}})=x_i(\vec{b})$ for all $i$ and
%$SW(\vec{\bar{b}})=SW(\vec{b})$; moreover, $\vec{\bar{b}}$ is a pure Nash
%equilibrium, because $\vec{b}$ is one, and the uniform price remains unchanged.
%
Let $\vec{\bar{b}}$ be a no-overbidding pure Nash equilibrium of the Uniform
Price Auction under uniform bidding, for arbitrary valuation functions. We argue
that it remains a pure Nash equilibrium under the standard bidding. If any
losing bidder $i$ has incentive to deviate using a {\em no-overbidding} standard
bid $\vec{b}_i$, so as to win at least one additional unit, he may do so also by
using a uniform bid $\langle b_i(1),1\rangle$, a contradiction. If a winning
bidder has an incentive to deviate under $\vec{\bar{b}}$ using a {\em
no-overbidding} standard bid $\vec{b}_i$, having $q_i$ non-zero components
(marginal bids), then he may as well do so with a uniform bid
$\langle\sum_{j\leq q_i}b_i(j)/q_i,q_i\rangle$; indeed, because $\vec{b}_i$ is
submodular, $\sum_{j\leq q_i}b_i(j)/q_i\geq b_i(q_i)$. Thus, the assumed uniform
bid should grant $i$ at least as many units as $\vec{b}_i$.
\end{proof}

The above Lemma, along with statement {\em(ii)} of Theorem~\ref{thm:main-sm}
from Section~\ref{section:sm}, leads to the following:

\begin{corollary}
The Price of Anarchy for pure Nash equilibria of the Uniform Price Auction with
submodular bidders under the standard or the uniform bidding interface is at
most $3.1462$.
\end{corollary}

%The lower bound follows directly from the first statement of
%Lemma~\ref{lemma:ub-sb} and the lower bound on the inefficiency of pure Nash
%equilibria in undominated strategies of~\cite{Markakis12} (Theorem 2). The
%upper bound follows from the second statement of Lemma~\ref{lemma:ub-sb} and
%statement {\em(ii)} of Theorem~\ref{thm:main-sm} (Section~\ref{section:sm}),
%which applies for mixed Bayes-Nash equilibria with no-overbidding, thus, also
%for pure Nash. 

For wider valuation function classes than submodular, we do not know whether the
Uniform Price Auction generally has pure Nash equilibria. To the best of our
knowledge, and as mentioned by Milgrom in~\cite{Milgrom04}, the standard
multi-unit auction formats have not been studied before, for any larger class of
valuation functions. In Section~\ref{section:sm} we give upper bounds of $4$ and
$6.2924$ on the Price of Anarchy of {\em mixed Bayes-Nash} equilibria for {\em
subadditive} valuation functions, under the {\em standard and uniform bidding
interfaces}, respectively. By Lemma~\ref{lemma:ub-sb} however, the former bound
of $4$ is valid also for uniform bidding, in the case of {\em pure Nash}
equilibria.

\begin{corollary}
The Price of Anarchy of pure Nash equilibria of the Uniform Price Auction with
subadditive bidders under the standard or the uniform bidding interface is at
most $4$.
\end{corollary}

\noindent The upper bound of this corollary follows by the second statement of
Lemma~\ref{lemma:ub-sb} and by Theorem~\ref{theorem:standard-sa-bounds},
discussed in Section~\ref{section:sm}. The lower bound is shown explicitly in
Section~\ref{subsection:lower-bounds}.

\subsection{Discriminatory Pricing}

The Discriminatory Auction, as a generalization of the First-Price Auction, is
not guaranteed to possess pure Nash equilibria; their existence depends heavily
on the choice of a tie-breaking rule, as is often the case for games where
players have a continuum of strategies. For example, consider the first-price
auction where the valuation of bidder $1$ is $1$, the valuation of bidder $2$ is
$\epsilon < 1$, and the tie-breaking rule always favors bidder $2$. Obviously
there can be no equilibrium where bidder $2$ bids above $1$. Furthermore, if
bidder $2$ bids some value $\delta < 1$, then bidder $1$ does not have a best
response in $(\delta, 1)$; no matter what he bids to win the unit, he always has
an incentive to lower his bid while still being above $\delta$. Hence there is
no Nash equilibrium for this auction. We exhibit here that, as with first-price
auctions, an appropriate choice of a tie-breaking rule induces a {\em uniform
bidding} profile that is a pure Nash equilibrium for the auction, even subject
to deviations under the standard bidding interface. Additionally, we show that
we can always obtain close approximations to pure Nash equilibria, i.e., pure
$\epsilon$-equilibria, for every possible tie breaking rule.

\begin{proposition}
\label{proposition:dpa-pne}
\mbox{} \\
\begin{compactenum}[(i)]
%%%
\item For every Discriminatory Auction there is a tie-breaking rule inducing a
uniform bidding profile that is a pure Nash equilibrium under that tie-breaking
rule.
\label{stmt:netiebreaking} 
%%%
\item For every $\epsilon > 0$, the Discriminatory Auction has a pure
$\epsilon$-equilibrium.
\label{stmt:epsilon-pne}
\end{compactenum}
%%%
\end{proposition}

\begin{proof}
(i) Let $\hat{\vec{m}}$ be the $nk$-dimensional vector obtained by appending all
vectors $\vec{m}_i, i \in [n]$, and let $\tilde{\vec{m}}$ be the vector obtained
by non-increasingly ordering $\hat{\vec{m}}$. We show that the set of bid
vectors $b$ where every bidder sets all his marginal bids to $\tilde{m}(k)$ is a
pure equilibrium, if ties are broken according to any tie-breaking rule that
satisfies $m_i(x_i(\vec{b})) \geq \tilde{m}(k)$.

Assume without loss of generality that there are at least $2$ players.  Let
$\vec{b}$ be the bidding profile where all players bid $\tilde{m}(k)$ on all
items, and break ties in any way that satisfies that $m_i(x_i(\vec{b})) \geq
\tilde{m}(k)$. To see why this is a pure equilibrium, consider the player
deviating to bid vector $\vec{b}_i'$. Note that $\vec{b}_i' - \vec{b}_i$ is
non-increasing. Define $\ell$ as the lowest index such that $(\vec{b}_i -
\vec{b}_i')(\ell)$ is negative (and define $\ell$ as $k+1$ if there is no such
index).  In case $\ell \leq x_i(\vec{b})$, then the utility of player $i$ will
certainly not increase by deviating to $\vec{b}_i'$, as he will lose utility
from the fact that $x_i(\vec{b}) - \ell$ less items are now allocated to him
under $(\vec{b}_i', \vec{b}_{-i})$, compared to $\vec{b}$. As player $i$ used to
derive non-negative utility from these items under $\vec{b}$, this removal of
items accounts for a non-negative decrease in utility. Moreover, player $i$
increases his bid on his first $\ell$ items, so this accounts for a non-negative
decrease in utility as well. His total utility will therefore decrease in this
case.

In case $\ell > x_i(\vec{b})$, we are in a situation where player $i$ increases
his bids (under $\vec{b}_i'$, compared to $\vec{b}_i$) on some of the first
$x_i(\vec{b})$ items by at least $0$, so he will win these items under
$(\vec{b}_i', \vec{b}_{-i})$ but spend more money on it, leading to a decrease
in utility. On any remaining items that player $i$ wins under $(\vec{b}_i',
\vec{b}_{-i})$, he overbids. This also accounts for a non-negative decrease in
utility. The total decrease in utility is thus non-negative in this case.

\medskip
(ii) Let $\vec{b}^*$ be a social welfare maximizing bidding profile. Consider
the uniform bidding profile $\vec{\bar{b}}$ as defined in
Proposition~\ref{proposition:dpa-pne}~(i). Let $(\xi_i)_{i \in N}$ be a set of
vectors that indicate the optimal allocation $(x_i(\vec{b}^*))_{i \in N}$, i.e.,
$\xi_i$ is the $(0,1)$-vector of which the first $x_i(\vec{b}^*)$ entries are
$1$, and the remaining entries are $0$. We show that $\tilde{\vec{b}} = \vec{b}
+ \epsilon\xi/k$ is a pure $\epsilon$-equilibrium. The reasoning we apply is
largely analogous to the proof of Proposition~\ref{proposition:dpa-pne}~(i). 

First of all, observe that there are no ties that need to be broken under
$\tilde{\vec{b}}$, and that the allocation $(\,x_i(\tilde{\vec{b}})\,)_{i \in
N}$ satisfies $m_i(x_i(\vec{b})) \geq \tilde{m}(k)$.

Consider the player deviating to bid vector $\vec{b}_i'$. If player $i$ wins
less items under $(\vec{b}_i', \tilde{\vec{b}}_{-i})$ than under $\vec{b}$, he
will experience an increase in utility of at most $(x_i(\tilde{\vec{b}}) -
x_i(\vec{b}'))\epsilon/k$ due to losing items, because the utility that player
$i$ derived under $\vec{b}'$ from each of these lost items was at least
$-\epsilon/k$. On the remaining items that player $i$ still wins, the player
increases his bid by at least $-\epsilon/k$, and this accounts for an increase
in utility of at most $x_i(\vec{b}')\epsilon/k$. The total increase in utility
is thus at most $x_i(\vec{b}')\epsilon/k \leq \epsilon$.

If player $i$ wins at least as much items under
$(\vec{b}_i',\tilde{\vec{b}}_{-i})$ than under $\tilde{\vec{b}}$, the player
will have increased his bids on the first $x_i(\tilde{\vec{b}})$ items by at
least $-\epsilon/k$, and by at least $0$ on the remaining items. For these
remaining items, the player experiences non-positive utility under $(\vec{b}_i',
\vec{b}_{-i})$, whereas he experienced $0$ utility under $\tilde{\vec{b}}$.
Therefore, the total increase in utility is in this case at most
$x_i(\tilde{\vec{b}})\epsilon/k \leq \epsilon$.
\end{proof}

We show next that, whenever pure Nash equilibria exist, they are socially
optimal, even with arbitrary valuation functions. This is in analogy with other
results on mechanisms with first price rules \cite{Hassidim11}.

\begin{theorem}
\label{thm:pne_discriminatory}
Pure Nash equilibria of the Discriminatory Auction (with the standard or the
uniform bidding interface) are always efficient, even for bidders with arbitrary
valuation functions.
\end{theorem}
 
The proof of Theorem \ref{thm:pne_discriminatory} is based on the following
Lemma, which captures the main properties of pure Nash equilibria. Notice that,
the first in the Lemma below, essentially states that every pure Nash
equilibrium of the auction occurs at a {\em uniform bidding} profile. Thus, the
theorem is also valid for the uniform bidding interface.

\begin{lemma}\label{lem:pne_properties}
Let $\vec{b}$ be a pure Nash equilibrium in a given Discriminatory Auction where
the bidders have general valuation functions. Let $d = \max \{b_i(j) : i \in
[n], j \in [k], j > x_i(\vec{b})\}$. Then:
%%%
\begin{enumerate}[(i)]
%%%
\item For any bidder $i$ who wins at least one item under $\vec{b}$, and for all
$j \in [x_i(\vec{b})]$: $b_i(j) = d$,
%%%
\item
%%%
$\displaystyle\ell d \leq \sum_{j = x_i(\vec{b}) - \ell + 1}^{x_i(\vec{b})} m_i(j),\,$
%%%
for all $i \in [n]$ and $\ell \in [x_i(\vec{b})]$,
%%%
\item
%%%
$\displaystyle\sum_{j = x_i(\vec{b}) + 1}^{x_i(\vec{b}) + \ell} m_i(j) \leq \ell d,\,$
%%%
for all $i \in [n]$ and $\ell \in [k - x_i(\vec{b})]$. 
%%%
\end{enumerate}
\end{lemma}
\begin{proof}
Let $c$ be the smallest value in $\{b_i(j) : i \in [n], j \in [k], j \leq
x_i(\vec{b})\}$, i.e., the smallest winning marginal bid.  Observe that $c=d$:
Otherwise, a player $i$ that bids $b_i(x_i(\vec{b})) = c$ could change
$b_i(x_i(\vec{b}))$ to a lower bid in order to obtain more utility.  For the
same reasons, we conclude that any winning marginal bid $b_i(j)$ is equal to the
largest marginal bid that is smaller than $b_i(j)$. It follows inductively that
all winning marginal bids are equal to $d$. This establishes point $(1)$ of the
claim.

Suppose that for some $i \in [n]$ and $\ell \in [x_i(\vec{b})]$, it holds that:
%%%
$$
\ell d
= 
\sum_{j = x_i(\vec{b}) - \ell + 1}^{x_i(\vec{b})} b_i(j) 
> 
\sum_{j = x_i(\vec{b}) - \ell + 1}^{x_i(\vec{b})} m_i(j).
$$
%%%
Then, if player $i$ changes all marginal bids $b_i(j)$ for $j \in \{j : \ell
\leq j\}$ to $0$, he would increase his utility. This is not possible since
$\vec{b}$ is a pure equilibrium, so we conclude that for all $i \in [n]$ and
$\ell \in [x_i(\vec{b})]$, it holds that:
%%%
$$
\sum_{j = x_i(\vec{b}) - \ell + 1}^{x_i(\vec{b})} b_i(j) 
\leq 
\sum_{j = x_i(\vec{b}) - \ell + 1}^{x_i(\vec{b})} m_i(j).
$$
%%%
This establishes point $(ii)$ of the claim. 

Assume that there is $i \in [n]$ and $\ell \in [k - x_i(\vec{b})]$ such that:
%%%
$$
\sum_{j = x_i(\vec{b}) + 1}^{x_i(\vec{b}) + \ell} m_i(j) > \ell d.
$$
%%%
Then, if player $i$ would change his marginal bids $b_i(j), j \in \{j : 1 \leq j \leq
x_i(\vec{b}) + \ell\}$ to $d + \epsilon$ for some $\epsilon > 0$, then player
$i$'s utility {\em increases by:}
%%%
$$
\sum_{j = x_i(\vec{b})+1}^{x_i(\vec{b})+\ell} m_i(j)
-\ell(d + \epsilon) - x_i(\vec{b})\epsilon.
$$ 
%%%
Because $\sum_{j = x_i(\vec{b})+1}^{x_i(\vec{b})+\ell} m_i(j) - \ell d$ is
positive, this total increase is positive when we take for $\epsilon$ a
sufficiently small value.  This is in contradiction with the fact that $\vec{b}$
is a pure equilibrium, and this establishes point $(iii)$ of the claim.
\end{proof}

\begin{proofof}{Theorem \ref{thm:pne_discriminatory}}
Let $\vec{b}^*$ be a bid vector that attains the optimum social welfare. Denote
by $A$ the set of bidders that get more items under $\vec{b}$ than under
$\vec{b}^*$. For a bidder $i \in A$, define $\ell_i$ as the number of extra
items that $i$ gets under $\vec{b}$, when compared to $\vec{b}^*$; i.e., $\ell_i
= x_i(\vec{b}) - x_i(\vec{b}^*)$. Denote by $B$ the set of bidders that get more
items under $\vec{b}^*$ than under $\vec{b}$. For a bidder $i \in B$, define
$\ell_i$ as the number of extra items that $i$ gets under $\vec{b}^*$, when
compared to $\vec{b}$; i.e., $\ell_i = x_i(\vec{b}^*) - x_i(\vec{b})$. Then,
%%%
\begin{align*}
& \sum_{i = 1}^n v_i(x_i(\vec{b})) - \sum_{i = 1}^n v_i(x_i(\vec{b}^*))
= \sum_{i = 1}^n \left(\sum_{j = 1}^{x_i(\vec{b})} m_i(j) - \sum_{j = 1}^{x_i(\vec{b}^*)} m_i(j) \right)\nonumber\\
&\qquad = \sum_{i \in A} \sum_{j = x_i(\vec{b}) - \ell_i + 1}^{x_i(\vec{b})} m_i(j) - \sum_{i \in B} \sum_{j = x_i(\vec{b}) + 1}^{x_i(\vec{b}) + \ell_i} m_i(j)
\ge \sum_{i \in A} \ell_i d - \sum_{i \in B} \ell_i d =0.
\end{align*}

\iffalse
\[
\begin{array}{llcl}
&\displaystyle
\sum_{i = 1}^n v_i(x_i(\vec{b})) - \sum_{i = 1}^n v_i(x_i(\vec{b}^*))
&=&\displaystyle
\sum_{i = 1}^n \left(\sum_{j = 1}^{x_i(\vec{b})} m_i(j) - \sum_{j = 1}^{x_i(\vec{b}^*)} m_i(j) \right)\nonumber\\
=&\displaystyle\sum_{i \in A} \sum_{j = x_i(\vec{b}) - \ell_i + 1}^{x_i(\vec{b})} m_i(j) - \sum_{i \in B} \sum_{j = x_i(\vec{b}) + 1}^{x_i(\vec{b}) + \ell_i} m_i(j)
&\geq&\displaystyle
\sum_{i \in A} \ell_i d - \sum_{i \in B} \ell_i d =0
\end{array}
\]
\fi
%%%
The inequality in the derivation above follows from points $(ii)$ and $(iii)$ of
Lemma \ref{lem:pne_properties}, and the final equality holds because $\sum_{i
\in A} \ell_i = \sum_{i \in B} \ell_i$.  Thus, the social welfare of the pure
equilibrium $\vec{b}$ is optimal.
\end{proofof}

%% file: bne.tex
\section{Bayes-Nash Inefficiency}
\label{section:sm}

\otrem{In this section we develop our main results, concerning the inefficiency
of (mixed) Bayes-Nash equilibria.}
%Our main results concern the inefficiency of Bayes-Nash equilibria.  For a
%discussion on the properties of pure Nash equilibria, we refer the reader to
%{\bf Appendix A}. 
We derive bounds on the (mixed) Bayesian Price of Anarchy for the Discriminatory
and the Uniform Price Auctions with submodular and subadditive valuation
functions. For the latter class \otrem{of valuation functions} our bounds are
the first results to appear in the literature of standard multi-unit auctions
(see also the commentary in~\cite[Chapter 7]{Milgrom04}). 
%Some proofs of this section are deferred to \textbf{Appendix B}. 

\begin{theorem}\label{thm:main-sm}
The \bpoa\ (under the standard or uniform bidding format) is at most 
%%%
\begin{compactenum}[(i)]
%%%
\item $\frac{e}{e-1}$ and $\frac{2e}{e-1}$ for the \dpa\ with submodular and
subadditive valuation functions, respectively,
%%%
\item  $|{\cal W}_{-1}(-1/e^2)|\approx 3.1462 < \frac{2e}{e-1}$ and $2|{\cal
W}_{-1}(-1/e^2)| \approx 6.2924 < \frac{4e}{e-1}$ for the \upa\ with submodular
and subadditive valuation functions, respectively, ${\cal W}_{-1}$ being the
lower branch of the Lambert
${\cal W}$ function.
%%%
\end{compactenum}
%%%
\end{theorem}

This theorem improves on the currently best known upper bounds of
$\frac{2e}{e-1}$ and $\frac{4e}{e-1}$ for the \dpa\ and the \upa, respectively,
with submodular valuation functions due to Syrgkanis and Tardos \cite{ST13}. For
the \upa, this further reduces the gap from the known lower bound of
$\frac{e}{e-1}$ \cite{Markakis12}. Syrgkanis and Tardos \cite{ST13} obtained
their bounds through an adaptation of the \emph{smoothness framework} for games
with incomplete information (\cite{Roughgarden12,Syrgkanis12a}). The bounds of
Theorem~\ref{thm:main-sm} and some additional results can also be obtained
through this framework. We comment on this in more detail in
Section~\ref{sec:z-smoothness}.

For subadditive valuation functions and the standard bidding format, however,
better bounds can be obtained by adapting a technique recently introduced by
Feldman {\em et al.}~\cite{Feldman12}, which does not fall within the smoothness
framework. We were unable to derive these bounds via a smoothness argument and
believe that this is due to the additional flexibility provided by this
technique. 

\begin{theorem}
\label{thm:main-sa}
The \bpoa\ is at most $2$ and $4$ for the \dpa\ and the \upa, respectively, with
subadditive valuation functions under the standard bidding format.
\label{theorem:standard-sa-bounds}
\end{theorem}

\subsection{Proof Template for Bayesian Price of Anarchy} 

In order to present all our bounds from Theorem \ref{thm:main-sm} and Theorem
\ref{thm:main-sa} in a self-contained and unified manner, we make use of a proof
template which is formalized in Theorem~\ref{theorem:generic-stmt} below.
Variants of this approach have been used in several previous works
(e.g.,~\cite{Markakis12,Christodoulou08,Bhawalkar11}). 

\begin{theorem}
\label{theorem:generic-stmt}
Let $V$ be a class of valuation functions. Suppose that for every valuation
profile $\vec{v} \in V^n$, for every bidder $i \in [n]$, and for every
distribution ${\cal P}_{-i}$ over non-overbidding profiles $\vec{b}_{-i}$, there
is a bidding profile $\vec{b}'_{i}$ such that the following inequality holds for
some $\lambda > 0$ and $\mu \ge 0$:
%%%
\begin{equation}
\mathbb{E}_{\vec{b}_{-i}\sim{\cal P}_{-i}}
\Bigl[
u_i^{\vec{v}_i}(\vec{b}'_i,\vec{b}_{-i})
\Bigr] 
\geq
\lambda\cdot v_i(x^\vec{v}_i)
- \mu \cdot \mathbb{E}_{\vec{b}_{-i}\sim{\cal P}_{-i}}
\Biggl[
\sum_{j=1}^{x_i^\vec{v}}\beta_j(\vec{b}_{-i})
\Biggr].
\label{equation:myFeldman-dpa}
\end{equation}
%%%
Then, the \bpoa\ is at most
%%%
\begin{compactenum}[(i)]
\item $\max\{1, \mu\}/\lambda$ for the \dpa, 
\item $(\mu+1)/\lambda$ for the \upa.
\end{compactenum}
%%%
\end{theorem}

\noindent Note that in this theorem we make no assumptions regarding the bidding
interface; proving a bound for the uniform bidding interface only requires that
we exhibit a uniform bidding strategy $\vec{b}'_i$ for each bidder $i$ and for
any distribution ${\cal P}_{-i}$ of uniform non-overbidding profiles
$\vec{b}_{-i}$. \otrem{The Theorem's proof follows.}

\begin{proof}
For any particular bidder $i\in [n]$ let ${\cal U}^i\subseteq{\cal V}$ denote
the subset of valuation profiles $\vec{v}\in{\cal V}$ where $x^{\vec{v}}(i) \geq
1$, i.e., ${\cal U}^i=\{\vec{v}\in{\cal V}|x^{\vec{v}}(i) \geq 1\}$; these are
the profiles under which $i$ is a ``socially optimum winner''. Accordingly, we
let ${\cal W}^{\vec{v}}$ denote the subset of ``socially optimum winners'' in
valuation profile $\vec{v}\in{\cal V}$. 

Consider a Bayes-Nash equilibrium $\vec{B}$. Fix any valuation profile
$\vec{v}=(\vec{v}_i,\vec{v}_{-i})\in{\cal V}$ and a bidder $i\in [n]$. Assume
that bidder $i$ deviates according to a bidding vector $\vec{b}'_i$ satisfying
\eqref{equation:myFeldman-dpa}, taking ${\cal P}_{-i} =
\vec{B}_{-i}^{\vec{w}_{-i}}, \vec{w}_{-i} \sim \pi_{-i}$. By taking expectation
over all valuation profiles $\vec{w}_{-i}\in{\cal V}_{-i}$, we obtain
%%%
\begin{align}
\underset{\vec{w}_{-i}|\vec{v}_i}{
\mathbb{E}}
\Biggl[
\underset{\vec{b}_{-i}\sim\vec{B}_{-i}^{\vec{w}_{-i}}}{
\mathbb{E}}
\Bigl[
u_i^{\vec{v}_i}(\vec{b}'_i,\vec{b}_{-i})\Bigr]\Biggr]
&\geq\displaystyle
\lambda v_i(x_i^{\vec{v}})-
\mu\cdot
\underset{\vec{w}_{-i}|\vec{v}_i}{
\mathbb{E}}\left[
\underset{\vec{b}_{-i}\sim\vec{B}_{-i}^{\vec{w}_{-i}}}{
\mathbb{E}}
\left[
\sum_{j=1}^{x_i^{\vec{v}}}\beta_j(\vec{b}_{-i})
\right]\right]\nonumber\\
& \geq
\lambda v_i(x_i^{\vec{v}})-
\mu\cdot
\underset{\vec{w}}{
\mathbb{E}}\left[
\underset{\vec{b}\sim\vec{B}^{\vec{w}}}{
\mathbb{E}}
\left[
\sum_{j=1}^{x_i^{\vec{v}}}\beta_j(\vec{b})\right]\right],\nonumber
\end{align}
%%%
where the last inequality holds because
$\beta_j(\vec{b}_{-i})\leq\beta_j(\vec{b})$ for every $j = 1, \dots, k$ and by
the independence of $\pi_i$, i.e.,
$\sum_{\vec{w}_{-i}}\pi(\vec{w}_{-i}|v_i)=1=\sum_{\vec{w}}\pi(\vec{w})$. Because
$\vec{B}$ is a Bayes-Nash equilibrium, bidder $i$ does not have an incentive to
deviate and, thus:
%%%
\[
\mathbb{E}_{\vec{w}_{-i}|\vec{v}_i}
\Biggl[
\mathbb{E}_{\vec{b}\sim\vec{B}^{(\vec{v}_i,\vec{w}_{-i})}}
\Bigl[
u_i^{\vec{v}_i}(\vec{b})
\Bigr]\Biggr]
\geq
\mathbb{E}_{\vec{w}_{-i}|\vec{v}_i}
\Biggl[
\mathbb{E}_{\vec{b}_{-i}\sim\vec{B}_{-i}^{\vec{w_{-i}}}}
\Bigl[
u_i^{\vec{v}_i}(\vec{b}'_i, \vec{b}_{-i})
\Bigr]\Biggr].
\]
%%%
We conclude that 
%%%
\[
\mathbb{E}_{\vec{w}_{-i}|\vec{v}_i}
\Biggl[
\mathbb{E}_{\vec{b}\sim\vec{B}^{(\vec{v}_i,\vec{w}_{-i})}}
\Big[
u_i^{\vec{v}_i}(\vec{b})
\Big]\Biggr]
+ \mu
\mathbb{E}_{\vec{w}}
\left[
\mathbb{E}_{\vec{b}\sim\vec{B}^{\vec{w}}}
\left[
\sum_{j=1}^{x_i^{\vec{v}}}\beta_j(\vec{b})
\right]\right]
\geq
\lambda v_i(x_i^{\vec{v}}).
\]

Taking expectation of both sides over the distribution of $\vec{v}\in{\cal U}^i$
and summing over all bidders we obtain
%%%
\begin{align}
&\displaystyle
\sum_{i \in [n]}
\sum_{\vec{v}\in{\cal U}^i}\pi(\vec{v})\cdot
\left(
\underset{\vec{w}_{-i}|\vec{v}_i}{
\mathbb{E}}\Biggl[
\underset{\vec{b}\sim\vec{B}^{(v_i,\vec{w}_{-i})}}{
\mathbb{E}}
\Big[u_i^{v_i}(\vec{b})\Big]\Biggr]%\nonumber\\
%%%%%%%%%%%%%%%%%%%%%%%%%%%%%%%%%%%%%%%%%%%%%%%%%%%%%%%%%%%
+
\mu
\underset{\vec{w}}{
\mathbb{E}}\left[
\underset{\vec{b}\sim\vec{B}^{\vec{w}}}{
\mathbb{E}}\left[
\sum_{j=1}^{x_i^{\vec{v}}}\beta_j(\vec{b})\right]\right]\right)
\nonumber\\
%%%%%%%%%%%%%%%%%%%%%%%%%%%%%%%%%%%%%%%%%%%%%%%%%%%%%%%%%%%
&\displaystyle\geq
\sum_{i \in [n]}
\sum_{\vec{v}\in{\cal U}^i}\pi(\vec{v})\cdot
\lambda v_i(x_i^{\vec{v}})
=\sum_{\vec{v}}\pi(\vec{v})\sum_{i\in{\cal W}^{\vec{v}}}
\lambda v_i(x_i^{\vec{v}})
=\lambda\mathbb{E}_{\vec{v}}\Bigl[SW(\vec{v},\vec{x^{\vec{v}}})\Bigr].
\nonumber %\label{equation:sa-bne-guarantee}
\end{align}

By standard manipulations, the latter simplifies to 
%%%
$$
\underset{\vec{v}}{
\mathbb{E}}\left[
\underset{\vec{b}\sim\vec{B}^{\vec{v}}}{
\mathbb{E}}
\left[
\sum_{i \in [n]} u_i^{\vec{v}_i}(\vec{b})
\right]\right]+
\mu
\underset{\vec{w}}{
\mathbb{E}}\left[
\underset{\vec{b}\sim\vec{B}^{\vec{w}}}{
\mathbb{E}}
\left[
\sum_{i \in [n]}\sum_{j=1}^{x_i^{\vec{v}}}\beta_{j}(\vec{b})
\right]\right]
\geq
\lambda
\underset{\vec{v}}{
\mathbb{E}}\Bigl[SW(\vec{v},\vec{x^{\vec{v}}})\Bigr].
$$
%%%
Note that $\sum_{i \in [n]} x^\vec{v}_i = k$ and that $\beta_j(\vec{b})$ is
non-decreasing in $j$. We can therefore bound 
%%%
$$
\sum_{i \in [n]}
\sum_{j = 1}^{x^\vec{v}_i}
\beta_j(\vec{b})
\le
\sum_{j = 1}^k
\beta_j(\vec{b})
$$
%%%
and obtain 
\begin{equation}
\label{eq:finalll}
\underset{\vec{v}}{
\mathbb{E}}\left[
\underset{\vec{b}\sim\vec{B}^{\vec{v}}}{
\mathbb{E}}
\left[
\sum_{i \in [n]} u_i^{\vec{v}_i}(\vec{b})
\right]\right]+ \mu
\underset{\vec{w}}{
\mathbb{E}}\left[
\underset{\vec{b}\sim\vec{B}^{\vec{w}}}{
\mathbb{E}}
\left[
\sum_{j=1}^{k} \beta_{j}(\vec{b})
\right]\right]
\geq
\lambda
\underset{\vec{v}}{
\mathbb{E}}\Bigl[
SW(\vec{v},\vec{x^{\vec{v}}})
\Bigr].
\end{equation}

Note that for the discriminatory pricing rule the total payments under $\vec{b}$
are equal to $\sum_{j=1}^{k} \beta_{j}(\vec{b})$. Thus \eqref{eq:finalll} yields
%%%
$$
\underset{\vec{v}}{
\mathbb{E}}\Biggl[
\underset{\vec{b}\sim\vec{B}^{\vec{v}}}{
\mathbb{E}}
\Bigl[
SW(\vec{v}, \vec{b})
\Bigr]\Biggr]+ (\mu - 1)
\underset{\vec{w}}{
\mathbb{E}}\left[
\underset{\vec{b}\sim\vec{B}^{\vec{w}}}{
\mathbb{E}}
\left[
\sum_{j=1}^{k} \beta_{j}(\vec{b})
\right]\right]
\geq
\lambda
\underset{\vec{v}}{
\mathbb{E}}\Bigl[SW(\vec{v},\vec{x^{\vec{v}}})\Bigr].
$$
%%%
If $\mu \le 1$, the first statement of the theorem holds. If $\mu > 1$, then we
exploit that the total payments satisfy $\sum_{j=1}^k\beta_j(\vec{b})\leq\sum_{i
\in [n]} v_i(x_i(\vec{b}))=SW(\vec{v},\vec{b})$ because players never
overbid.\footnote{In order to derive the claimed bound for the case $\mu > 1$ we
need to exploit the no-overbidding assumption mentioned in
Section~\ref{section:definitions} for the Discriminatory Auction as well.
However, all our bounds derived in this paper exploit only that $\mu \le 1$ and
thus hold even without this assumption.} Dividing both sides by $\mu > 0$ proves
the first statement of the theorem in this case. 

For the uniform price rule we use that
$\sum_{j=1}^k\beta_j(\vec{b})\leq\sum_{i \in [n]}
v_i(x_i(\vec{b}))=SW(\vec{v},\vec{b})$ be\-cau\-se of the no-overbidding assumption
and that $\sum_{i \in [n]} u_i^{\vec{v}_i}(\vec{b})\leq\sum_{i \in [n]}
v_i(x_i(\vec{b}))=SW(\vec{v},\vec{b})$.  Thus, \eqref{eq:finalll} yields 
%%%
$$
(\mu + 1) \mathbb{E}_{\vec{v}}\left[\mathbb{E}_{\vec{b}\sim\vec{B}^{\vec{v}}}
\left[
SW(\vec{v}, \vec{b})
\right]\right]
\geq
\lambda\mathbb{E}_{\vec{v}}[SW(\vec{v},\vec{x^{\vec{v}}})].
$$
%%%
Dividing both sides by $\mu+1 > 0$ proves the second statement of the theorem.
\end{proof}

In Section~\ref{subsection:lower-bounds} we show that our bound of
$\frac{e}{e-1}$ for the \dpa\ is essentially best possible, if one sticks to the
proof template of Theorem \ref{theorem:generic-stmt}; this also rules out that
better bounds can be obtained via the techniques described in \cite{ST13} or
even \cite{Feldman12}\otrem{; the latter we use in
subsection~\ref{subsection:sa-standard}.}

\subsection{
Key Lemma and Proofs of Theorem~\ref{thm:main-sm} and Theorem~\ref{thm:main-sa}
}

The following is our key lemma to prove Theorem~\ref{thm:main-sm}. We point out
that it applies to arbitrary valuation functions and to any multi-unit auction
which is \emph{discriminatory price dominated}, i.e., the total payment
$P_i(\vec{b})$ of bidder $i$ under profile $\vec{b}$ satisfies $P_i(\vec{b}) \le
\sum_{j \in [x_i(\vec{b})]} b_i(j)$. Note that every multi-unit auction
guaranteeing \emph{individual rationality} must satisfy this condition.

\begin{lemma}[Key Lemma]
\label{lem:smoothness-key}
Let $\vec{v}$ be a valuation profile and suppose that the pricing rule is
discriminatory price dominated. Define $\tau_i = \arg \min_{j \in
[x_i^{\vec{v}}]} v_i(j)/j$ for every $i \in [n]$.  Then for every bidder $i \in
[n]$ and every bidding profile $\vec{b}_{-i}$ 
%there exists a randomized uniform bidding profile $\vec{b}'_i$ 
there exists a randomized uniform bidding strategy $B_i'$, such that for every
$\alpha > 0$
%%%
\begin{equation}
\label{eq:rand-smooth}
\mathbb{E}_{\vec{b}_i'\sim B_i'}\Bigl[
u_i^{\vec{v}_i}(\vec{b}'_i, \vec{b}_{-i})
\Bigr]\ge 
\alpha \left(1 - \frac{1}{e^{1/\alpha}}\right)
x_i^{\vec{v}} \frac{v_i(\tau_i)}{\tau_i} - 
\alpha \sum_{j = 1}^{x^\vec{v}_i}\beta_j(\vec{b}_{-i}).
\end{equation}
\end{lemma}

\begin{proof}
%%%
Define $G=(1-e^{-1/\alpha})$ and let $\vec{c}_i$ be the vector that is
$v_i(\tau_i)/\tau_i$ on the first $x_i^{\vec{v}}$ entries, and is $0$ everywhere
else. Let $t$ be a random variable drawn from $[0, G]$ with probability density
function $f(t) = \alpha/(1 - t)$. Define the random deviation of bidder $i$ as
$\vec{b}_i' = t \vec{c}_i$. Note that $\vec{b}'_i$ is always a uniform bid
vector. 

Let $k^*$ be the number of items that bidder $i$ would win under profile
$(G\vec{c}_i,\vec{b}_{-i})$, i.e., the number of items won by $i$, when $i$
would deviate to bid vector $G \vec{c}_i$. For $j = 0, \dots, k^*$, let
$\gamma_j$ refer to the infimum value in $[0,G]$ such that bidder $i$ would win
$j$ items if he would deviate to bid vector $\gamma_j \vec{c}_i$. Note that this
definition is equivalent to defining $\gamma_j$ as the least value in $[0,G]$
that satisfies $\gamma_j v_i(\tau_i)/{\tau_i} = \beta_{j}(\vec{b}_{-i})$.  For
notational convenience, we define $\gamma_{k^*+1} = G$.

Let $x_i(\vec{b}'_i, \vec{b}_{-i})$ be the random variable that denotes the
number of units allocated to bidder $i$ under $(\vec{b}'_i, \vec{b}_{-i})$. It
always holds that $x_i(\vec{b}'_i, \vec{b}_{-i}) \le k^* \leq x_i^{\vec{v}}$,
because bidder $i$ bids $b'_i(j) = 0$ for all $j = x^\vec{v}_i + 1, \dots, k$.
More precisely, we have $x_i(\vec{b}'_i, \vec{b}_{-i}) = j$ if $t \in (\gamma_j,
\gamma_{j+1}]$ for $j = 0, \dots, k^*$. By assumption, the payment of bidder $i$
under profile $(\vec{b}_i', \vec{b}_{-i})$ is at most $t x_i(\vec{b}_i',
\vec{b}_{-i}) v_i(\tau_i)/\tau_i$. Also note that, by definition of $\tau_i$, it
holds that $v_i(j) \geq j v_i(\tau_i)/\tau_i$ for $j \leq x_i^{\vec{v}}$. Using
these two facts, we can bound the expected utility of bidder $i$ as follows:
%%%
\begin{align*}
\mathbb{E}_{\vec{b}_i'\sim B_i'}\Bigl[
u_i^{\vec{v}_i}(\vec{b}'_i, \vec{b}_{-i})
\bigr] 
& \geq \sum_{j = 1}^{k^*} \int_{\gamma_{j}}^{\gamma_{j+1}} \Big(v_i(j) - t j \frac{v_i(\tau_i)}{\tau_i}\Big) f(t) dt \\
& \geq \sum_{j = 1}^{k^*} \int_{\gamma_{j}}^{\gamma_{j+1}} j \frac{v_i(\tau_i)}{\tau_i} (1 - t) f(t) dt = \alpha \sum_{j = 1}^{k^*} j \frac{v_i(\tau_i)}{\tau_i} \int_{\gamma_{j}}^{\gamma_{j+1}} 1 dt \\
& = \alpha \sum_{j = 1}^{k^*} j \frac{v_i(\tau_i)}{\tau_i} (\gamma_{j+1} - \gamma_{j}) 
 = \alpha G k^* \frac{v_i(\tau_i)}{\tau_i} - \alpha \sum_{j = 1}^{k^*} \gamma_j \frac{v_i(\tau_i)}{\tau_i}  \\
& = \alpha G k^* \frac{v_i(\tau_i)}{\tau_i} - \alpha \sum_{j = 1}^{k^*} \beta_j(\vec{b}_{-i}) \geq \alpha G x_i^{\vec{v}} \frac{v_i(\tau_i)}{\tau_i} - \alpha \sum_{j = 1}^{x_i^{\vec{v}}} \beta_j(\vec{b}_{-i}).
\end{align*}
%%%
The last inequality holds because $G v_i(\tau_i)/\tau_i \leq
\beta_{j}(\vec{b}_{-i})$, for $k^* + 1 \leq j \leq x_i^{\vec{v}}$, by the
definition of $k^*$. The above derivation implies (\ref{eq:rand-smooth}).
\end{proof}

\otrem{
The deviation $B_i'$, defined in Lemma~\ref{lem:smoothness-key}, is a
distribution on pure uniform bidding strategies $\vec{b}'_i$. Because the lemma
holds {\em for every} (pure) bidding profile $\vec{b}_{-i}$ of bidders other
than $i$, we can take expectation of both sides of \eqref{eq:rand-smooth}, over
any distribution ${\cal P}_{-i}$ of such profiles $\vec{b}_{-i}$. Then, we
obtain a version of the inequality,~\eqref{equation:myFeldman-dpa}, required by
Theorem~\ref{theorem:generic-stmt}, with expected values on both sides, over the
distribution $B_i'$. This means that $B_i'$ contains at least one pure strategy
$\vec{b}_i'$ in its support, satisfying exactly \eqref{equation:myFeldman-dpa},
for $(\lambda,\mu) = (\alpha \left(1 - e^{-1/\alpha}\right),\alpha)$.}

\otrem{Moreoever, because $B_i'$ yields pure uniform bidding strategies, the
lemma applies to both the standard and the uniform bidding interfaces.  Observe
also that every $\vec{b}'_i\sim B_i'$ satisfies the no-overbidding assumption.
}

\begin{proofof}{of Theorem~\ref{thm:main-sm}}
First consider the case of submodular valuation functions. In this case, $\tau_i
= x_i^{\vec{v}}$ for every $i \in [n]$, as explained in Section
\ref{section:definitions}. Using our Key Lemma, we conclude that
Theorem~\ref{theorem:generic-stmt} holds for $(\lambda,\mu) = (\alpha \left(1 -
e^{-1/\alpha}\right),\alpha)$. The stated bounds are obtained by choosing
$\alpha = 1$ for the \dpa\ and $\alpha = -1/(W_{-1}(-1/e^2) + 2) \approx 0.87$
for the \upa. 

Next consider the case of subadditive valuation functions. The following lemma
shows that subadditive valuation functions can be approximated by uniform ones,
thereby losing at most a factor $2$. 

\begin{lemma}\label{lem:approx}
If $v_i$ is subadditive, then $\frac{v_i(\tau_i)}{\tau_i} \ge \frac12
\frac{v_i(x_i^\vec{v})}{x_i^\vec{v}}$ with $\tau_i = \arg \min_{j \in
[x_i^{\vec{v}}]} \frac{v_i(j)}{j}$.
\end{lemma}

By combining Lemma~\ref{lem:approx} with our Key Lemma, it follows that
Theorem~\ref{theorem:generic-stmt} holds for
$(\lambda,\mu)=(\frac{\alpha}{2}\left(1 - e^{-1/\alpha}\right),\alpha)$. The
bounds stated in Theorem~\ref{thm:main-sm} are obtained by the same choices of
$\alpha$ as for the submodular valuation functions.  \end{proofof}

\begin{proofof}{Lemma~\ref{lem:approx}}
Recall that $\tau_i = \arg \min_{j \in [x_i^{\vec{v}}]} v_i(j)/j$. Thus, it
suffices to show that $v_i(\tau_i)/\tau_i\geq v_i(x_i^\vec{v})/(2x_i^\vec{v})$.
By subadditivity, $v_i(x_i^\vec{v})/(2x_i^\vec{v})$ is at most:
%%%
\[
\frac{1}{2}\left(
\frac{v_i(x_i^\vec{v}-\tau_i)}{x_i^\vec{v}}+
\frac{v_i(\tau_i)}{x_i^\vec{v}}\right)\leq
\left\{
\begin{array}{ll}
\frac{1}{2}\left(
\frac{v_i(x_i^\vec{v}-\tau_i)}{\tau_i}+
\frac{v_i(\tau_i)}{\tau_i}\right)\leq
\frac{v_i(\tau_i)}{\tau_i} 
& \mbox{\ if\ }x_i^\vec{v}\leq 2\tau_i\\
\frac{1}{2}\left(
\frac{v_i(x_i^\vec{v}-\tau_i)}{x_i^\vec{v}}+
\frac{v_i(\tau_i)}{\tau_i}\right)\leq
\frac{v_i(\tau_i)}{\tau_i} 
& \mbox{\ if\ }x_i^\vec{v} > 2\tau_i
\end{array}\right.
\]

\noindent where in the first case 
%($x_i^\vec{v}-\tau_i\leq \tau_i$) 
we used monotonicity, i.e., $v_i(x_i^\vec{v}-\tau_i)\leq v_i(\tau_i)$ and in the
second case
%($x_i^\vec{v}-\tau_i > \tau_i$) 
we used subadditivity,
%~(\ref{sa-property})
i.e.,
$\frac{v_i(\tau_i)}{\tau_i}\geq\frac{v_i(x_i^\vec{v}-\tau_i)}{x_i^\vec{v}-\tau_i
+ \tau_i}=\frac{v_i(x_i^\vec{v}-\tau_i)}{x_i^\vec{v}}$.
\end{proofof}

\subsection{Subadditive Valuation Functions}
\label{subsection:sa-standard}

%\noindent 
Next, consider subadditive valuations under the \emph{standard} bidding format.
We derive improved bounds of $2$ and $4$ for the Discriminatory and Uniform
Price Auction, respectively. To this end, we adapt an approach recently
developed by Feldman {\em et al.}~\cite{Feldman12} to establish an analog of our
Key Lemma. The main idea is to construct the bid $\vec{b}_i'$ by using the
distribution ${\cal P}_{-i}$ on the profiles $\vec{b}_{-i}$.
Theorem~\ref{theorem:standard-sa-bounds} then follows from
Theorem~\ref{theorem:generic-stmt} in combination with
Lemma~\ref{lemma:myFeldman-dpa} below.

\begin{lemma}
Let $V$ be the class of subadditive valuation functions.  Then Theorem~3 holds
true with $(\lambda,\mu)=(\frac{1}{2},1)$ for the Discriminatory and
$(\lambda,\mu)=(\frac{1}{2},1)$ for the Uniform Price Auction (under the
standard bidding format).
%%%
\label{lemma:myFeldman-dpa}
\end{lemma}

\subsubsection*{Discriminatory Pricing}

We consider first the discriminatory pricing rule. We shall re-use certain
arguments from the proof for this case, in the proof for uniform pricing. Fix
any bidder $i$ and let $\vec{b}_i$ be a bidding vector, conforming to the
requirement of non-increasing marginal bids and having only its first $x$
components equal to a non-zero value. Given any bidding profile $\vec{b}_{-i}$:
%%%
$$
u_i^{v_i}(\vec{b}_i,\vec{b}_{-i})\geq 
v_i\Bigl(x_i(\vec{b}_i,\vec{b}_{-i})\Bigr)-
\sum_{j\leq x}b_i(j),
$$
%%%
because $i$ may pay at most $\sum_{j\leq k}b_i(j)=\sum_{j\leq x}b_i(j)$, by the
definition of $\vec{b}_i$ and the auction's payment rule. Taking expectation
over the distribution ${\cal P}$ of $\vec{b}_{-i}$ we have:
%%%
\begin{equation}
\mathbb{E}_{\vec{b}_{-i}\sim{\cal P}}\Bigl[
u_i^{v_i}(\vec{b}_i,\vec{b}_{-i})
\Bigr]
\geq 
\mathbb{E}_{\vec{b}_{-i}\sim{\cal P}}\Biggl[
v_i\Bigl(x_i(\vec{b}_i,\vec{b}_{-i})\Bigr)\Biggr]
-
\sum_{j\leq x}b_i(j)
\label{equation:post-feldman-dpa}
\end{equation}

From this point on, we analyze the right-hand side
of~(\ref{equation:post-feldman-dpa}). Given the distribution ${\cal P}$ of
$\vec{b}_{-i}$, let ${\cal D}$ denote the distribution of the $k$ highest bids
in $\vec{b}_{-i}$, $\beta_1(\vec{b}_{-i})\leq\cdots\leq\beta_k(\vec{b}_{-i})$;
to simplify notation, in the sequel we use simply $\beta_j$, $j=1,\dots, k$,
without a reference to $\vec{b}_{-i}$. For every fixed bid vector of bidder $i$,
the expected utility of  $i$ when the other bidders bid according to ${\cal P}$,
is equal to the expected utility of bidder $i$ in the two-bidders auction, where
the other bidder bids according to ${\cal D}$. We can thus assume that $i$
competes only against $\beta\sim {\cal D}$, containing $\beta_j$, $j=1,\dots,
k$. For notational purposes below, we shall also use $\gamma$ to denote a vector
drawn from ${\cal D}$. Notice that each $\beta\sim{\cal D}$ contains $\beta_j$,
$j=1,\dots,k$ in non-increasing order, as required by the auction format;
however, for our analysis, we find it convenient to index the individual bids,
$\beta_j$, in non-decreasing order.

We consider what happens when $i$ responds to ${\cal P}$ (i.e., ${\cal D}$, in
the two-bidders auction), by bidding a $\vec{b}_i=\tilde{\beta}$, that he
constructs as follows: he samples a vector $\beta$ from ${\cal D}$ and {\em
zeroes-out} the $k-x$ {\em highest} values in $\beta$. Subsequently, he adds to
all components of the ``truncated'' vector a sufficiently small $\epsilon$. This
latter modification we shall omit in our analysis below, to simplify its
exposition. It serves the purpose of avoiding ties with the opposing bids; the
analysis is valid when this $\epsilon$ becomes vanishingly small. Let
$\tilde{\cal D}$ denote the distribution of such $\vec{b}_i$. Continuing
from~(\ref{equation:post-feldman-dpa}), the expected utility of $i$ over
$\vec{b}_i\sim\tilde{\cal D}$ is: 
%%%
\begin{align}
%%%%%%%%%%%%%%%%%%%%%%%%%%%%%%%%%%%%%%%%%%%%%%%%%%%%%%%%%%%
&\displaystyle
\underset{\vec{b}_i\sim\tilde{\cal D}}{
\mathbb{E}}\Biggl[
\underset{\vec{b}_{-i}\sim{\cal P}}{
\mathbb{E}}\Bigl[
u_i^{v_i}(\vec{b}_i,\vec{b}_{-i})
\Bigr]\Biggr]
\geq 
\underset{\vec{b}_i\sim\tilde{\cal D}}{
\mathbb{E}}\Biggl[
\underset{\vec{b}_{-i}\sim{\cal P}}{
\mathbb{E}}\Biggl[
v_i\Bigl(x_i(\vec{b}_i,\vec{b}_{-i})\Bigr)\Biggr]
-
\sum_{j\leq x}b_i(j)
\Biggr]\nonumber\\
%%%%%%%%%%%%%%%%%%%%%%%%%%%%%%%%%%%%%%%%%%%%%%%%%%%%%%%%%%%
&=\displaystyle
\underset{\substack{\beta\sim{\cal D}\\ \vec{b}_{-i} \sim {\cal P}}}{
\mathbb{E}}
\Biggl[
v_i\Bigl(x_i(\tilde\beta,\vec{b}_{-i})\Bigr)
-
\sum_{j\leq x}\beta_j
\Biggr]
%%%%%%%%%%%%%%%%%%%%%%%%%%%%%%%%%%%%%%%%%%%%%%%%%%%%%%%%%%%
=\displaystyle
\underset{\substack{\beta\sim{\cal D}\\ \gamma\sim{\cal D}}}{
\mathbb{E}}
\Biggl[
v_i\Bigl(x_i(\tilde\beta,\gamma)\Bigr)
\Biggr]
-
\underset{\substack{\beta\sim{\cal D}\\ \gamma\sim{\cal D}}}{
\mathbb{E}}
\Biggl[
\sum_{j\leq x}\beta_j
\Biggr]%\nonumber
\label{equation:feldman-trick-begin}\\
%%%%%%%%%%%%%%%%%%%%%%%%%%%%%%%%%%%%%%%%%%%%%%%%%%%%%%%%%%%
&=\displaystyle
\frac{1}{2}\cdot
2\cdot
\underset{\substack{\beta\sim{\cal D}\\ \gamma\sim{\cal D}}}{
\mathbb{E}}
\Biggl[
v_i\Bigl(x_i(\tilde\beta,\gamma)\Bigr)
\Biggr]
-
\underset{\beta\sim{\cal D}}{
\mathbb{E}}
\left[
\sum_{j\leq x}\beta_j
\right]\nonumber\\
%%%%%%%%%%%%%%%%%%%%%%%%%%%%%%%%%%%%%%%%%%%%%%%%%%%%%%%%%%%
&=\displaystyle
\frac{1}{2}\left\{
\underset{\substack{\beta\sim{\cal D}\\ \gamma\sim{\cal D}}}{
\mathbb{E}}%_{\substack{\beta\sim{\cal D}\\ \gamma\sim{\cal D}}}
\Biggl[
v_i\Bigl(x_i(\tilde\beta,\gamma)\Bigr)
+
v_i\Bigl(x_i(\tilde\gamma,\beta)\Bigr)
\Biggr]
\right\}
-
\underset{\beta\sim{\cal D}}{
\mathbb{E}}
\left[
\sum_{j\leq x}\beta_j
\right]\nonumber\\
%%%%%%%%%%%%%%%%%%%%%%%%%%%%%%%%%%%%%%%%%%%%%%%%%%%%%%%%%%%
&\geq\displaystyle
\frac{1}{2}
\mathbb{E}_{\substack{\beta\sim{\cal D}\\ \gamma\sim{\cal D}}}
\Bigl[v_i(x)\Bigr]
-
\mathbb{E}_{\beta\sim{\cal D}}
\left[
\sum_{j\leq x}\beta_j
\right]
%%%%%%%%%%%%%%%%%%%%%%%%%%%%%%%%%%%%%%%%%%%%%%%%%%%%%%%%%%%
=\displaystyle
\frac{1}{2}v_i(x)
-
\mathbb{E}_{\vec{b}_{-i}\sim{\cal P}}
\left[
\sum_{j=1}^x\beta_j(\vec{b}_{-i})
\right]
\label{equation:feldman-trick-end}
\end{align}
%%%
where~(\ref{equation:feldman-trick-begin}) is due to the fact that $\sum_{j\leq
x}b_i(j)=\sum_{j\leq x}\tilde\beta_j=\sum_{j\leq x}\beta_j$, by construction of
$\tilde\beta$ and because $\vec{b}_i$ has its components (taken from
$\tilde\beta$) in non-increasing order. The last inequality,
in~(\ref{equation:feldman-trick-end}), holds by subadditivity of $v_i$,
particularly, because $x_i(\beta,\gamma)+x_i(\gamma,\beta)\geq x$ (when ties are
always resolved in favor of $i$). \otrem{From the derivation above we deduce
that there is at least one pure bidding strategy $\vec{b}_i\sim\bar{{\cal D}}$
that satisfies the requirements of Theorem~\ref{theorem:generic-stmt}
and~\eqref{equation:myFeldman-dpa}, with $(\lambda, \mu)=\left(\frac{1}{2},
1\right)$.}\qed

\subsubsection*{Uniform Pricing}

Fix any bidder $i$ and let $\vec{b}_i$ be a bidding vector with non-zero value
for each of the first $x$ components and zero value for the rest. Notice that,
given any bidding profile $\vec{b}_{-i}$:

\begin{align}
u_i^{v_i}(\vec{b}_i,\vec{b}_{-i})
&=\displaystyle
v_i\Bigl(x_i(\vec{b}_i,\vec{b}_{-i})\Bigr)-
x_i(\vec{b}_i,\vec{b}_{-i})\cdot p(\vec{b}_i,\vec{b}_{-i})
\nonumber\\
&\geq\displaystyle
v_i\Bigl(x_i(\vec{b}_i,\vec{b}_{-i})\Bigr)-
\sum_{j\leq x}b_i(j)\nonumber
\end{align}

\noindent where $p(\vec{b}_i,\vec{b}_{-i})$ is the uniform price that $i$ will
pay under the profile $(\vec{b}_i,\vec{b}_{-i})$. This price cannot be more than
the sum of the winning bids, $\sum_{j\leq x}b_i(j)$; this justifies the
inequality above. Taking expectation over the distribution ${\cal P}$ of
$\vec{b}_{-i}$ we have:

\begin{equation}
\mathbb{E}_{\vec{b}_{-i}\sim{\cal P}}\Bigl[
u_i^{v_i}(\vec{b}_i,\vec{b}_{-i})
\Bigr]
\geq 
\mathbb{E}_{\vec{b}_{-i}\sim{\cal P}}\Biggl[
v_i\Bigl(x_i(\vec{b}_i,\vec{b}_{-i})\Bigr)\Biggr]
-
\sum_{j\leq x}b_i(j)
\label{equation:post-feldman-upa}
\end{equation}

Our analysis from this point on will focus on identifying an appropriate bid
$\vec{b}_i$ for bidder $i$, that satisfies also {\em no-overbidding}.
Subsequently, we shall return to process~(\ref{equation:post-feldman-upa})
further. As previously, given the distribution ${\cal P}$ of $\vec{b}_{-i}$, we
define the distribution ${\cal D}$ of the $k$ highest bids in $\vec{b}_{-i}$,
$\beta_1(\vec{b}_{-i})\leq\cdots\leq\beta_k(\vec{b}_{-i})$; we use simply
$\beta_j$, $j=1,\dots, k$, without a reference to $\vec{b}_{-i}$. Each of these
is a potential uniform price that $i$ might need to pay, depending on his own
bids, $\vec{b}_i$. Moreover, we can assume that $i$ competes against a bid vector
$\beta$, containing only the bids $\beta_j$, $j=1,\dots,k$, and need not bother
with any smaller bids. Fix any such vector $\beta$ from the support of ${\cal
D}$ and let $T_{\beta}\subseteq[x]$ be a maximal subset of indices, such that:
$\sum_{j\in T_{\beta}}\beta_j > v_i(|T_{\beta}|)$. Let
$\bar{T}_{\beta}=[x]\setminus T_{\beta} $. Then, {\em we claim that:}
$
\sum_{j\in \bar{T}_{\beta}}\beta_j\leq v_i(|\bar{T}_{\beta}|)
$.
%%%
\noindent Indeed, if there exists $R\subseteq \bar{T}_{\beta}$ with $\sum_{j\in
R}\beta_j > v_i(|R|)$, then, by subadditivity of $v_i$ and monotonicity:

$$
v_i(|R\cup T_{\beta}|)\leq v_i(|R|)+v_i(|T_{\beta}|)<
\sum_{j\in R}\beta_j+\sum_{j\in T_{\beta}}\beta_j=
\sum_{j\in R\cup T_{\beta}}\beta_j
$$

\noindent which contradicts the maximality of $T_{\beta}$. 

Next, define $\tilde{\cal D}$ to be a distribution of bidding vectors similar to
${\cal D}$, differing from it as follows. For every vector $\beta$ in the
support of ${\cal D}$ that occurs with a certain probability, a vector
$\tilde\beta$ exists in the support of $\tilde{\cal D}$, that occurs with the
same probability and is made from $\beta$ according to the following rules: 
%%%
\begin{enumerate}
\item Identify a subset of indices $T_{\beta}$ for $\beta$ and let
$\bar{T}_{\beta}=[x]\setminus T_{\beta}$.
%%%
\item $\tilde\beta$ contains all bids of $\beta$ except for $\{\beta_j|j\in
T_{\beta}\}$; it contains a zero for each bid therein.
%%%
\end{enumerate}
%%%
%As previously, in our analysis for the Discriminatory Auction, a further modification
%that we shall omit in our analysis is that: $i$ adds a sufficiently small
%$\epsilon>0$ to each of the components of the resulting vector. This serves
%tie-breaking in favor of $i$, as before. 
Sampling a vector from $\tilde{\cal D}$ is equivalent to sampling a vector
$\beta$ from ${\cal D}$ and constructing $\tilde\beta$ as prescribed (using ``\
$\tilde{}$\ '' over a $\beta$ will correspond to this processing). For any
$\beta\sim{\cal D}$ and for any arbitrary bid $\vec{b}_i$, we observe that
$x_i(\vec{b}_i, \tilde\beta)\leq x_i(\vec{b}_i,\beta)+|T_{\beta}|$. Thus, by
subadditivity and monotonicity of $v_i$:
%%%
$$
v_i\Bigl(x_i(\vec{b}_i, \tilde\beta)\Bigr)
\leq
v_i\Bigl(x_i(\vec{b}_i,\beta)\Bigr)
+
v_i(\,|T_{\beta}|\,)
$$
%%%
Using $\sum_{j\leq x}\beta_j-\sum_{j\in \bar{T}_{\beta}}\beta_j=\sum_{j\in
T_{\beta}}\beta_j\geq v_i(\,|T_{\beta}|\,)$, we obtain:

\begin{equation}
v_i\Bigl(x_i(\vec{b}_i, \beta)\Bigr)-
\sum_{j\in \bar{T}_{\beta}}\beta_j
\geq
v_i\Bigl(x_i(\vec{b}_i, \tilde\beta)\Bigr)-
\sum_{1\leq j\leq x}\beta_j
\label{equation:feldman-1}
\end{equation}

\begin{equation}
\mbox{\noindent Thus:\quad}
\underset{\beta\sim{\cal D}}{
\mathbb{E}}\left[
v_i\Bigl(x_i(\vec{b}_i, \beta)\Bigr)-
\sum_{j\in \bar{T}_{\beta}}\beta_j
\right]
\geq
\underset{\beta\sim{\cal D}}{
\mathbb{E}}\left[
v_i\Bigl(x_i(\vec{b}_i, \tilde\beta)\Bigr)-
\sum_{1\leq j\leq x}\beta_j
\right]
\label{equation:feldman-2}
\end{equation}

Now consider $\vec{b}_i$ being drawn from the distribution $\tilde{\cal D}$.
Notice that, by their construction, all the bid vectors in the support of
$\tilde{\cal D}$ satisfy no-overbidding. As previously, in our analysis for the
Discriminatory Auction, a further modification that we shall omit in the
exposition of our analysis is that: $i$ adds a sufficiently small $\epsilon>0$
to each of the components of the bidding vector $\vec{b}_i$ drawn from
$\tilde{\cal D}$ (even to the zero-valued ones). We take expectation of the {\em
left-hand side} of~(\ref{equation:feldman-2}) over $\vec{b}_i\sim\tilde{\cal
D}$:

\begin{align}
\displaystyle
\underset{\substack{
\vec{b}_i\sim\tilde{\cal D}\\
\gamma\sim{\cal D}}
}{\mathbb{E}}
\left[
v_i\Bigl(x_i(\vec{b}_i, \gamma)\Bigr)
-\sum_{j\in \bar{T}_{\gamma}}\gamma_j
\right]
%%%%%%%%%%%%%%%%%%%%%%%%%%%%%%%%%%%%%%%%%%%%%%%%
&
=\displaystyle
\underset{\substack{
\beta\sim{\cal D}\\ 
\gamma\sim{\cal D}}
}{\mathbb{E}}\left[
v_i\Bigl(x_i(\tilde{\beta}, \gamma)\Bigr)
-
\sum_{j\in \bar{T}_{\gamma}}\gamma_j
\right]\nonumber\\
%%%%%%%%%%%%%%%%%%%%%%%%%%%%%%%%%%%%%%%%%%%%%%%%
&
=\displaystyle
\underset{\substack{
\beta\sim{\cal D}\\ \gamma\sim{\cal D}}
}{\mathbb{E}}\Biggl[
v_i\Bigl(x_i(\tilde\beta, \gamma)\Bigr)
\Biggr]
-
\underset{\beta\sim{\cal D}}{\mathbb{E}}
\left[
\sum_{j\in \bar{T}_{\beta}}\beta_j
\right]
\label{equation:feldman-half-1}
\end{align}

\noindent Accordingly, we take expectation of the {\em right-hand side}
of~(\ref{equation:feldman-2}) over $\bar{{\cal D}}$, to obtain:

\begin{align}
\displaystyle
\underset{\substack{
\vec{b}_i\sim\bar{{\cal D}}\\ 
\gamma\sim{\cal D}}
}{\mathbb{E}}
\left[
v_i\Bigl(x_i(\vec{b}_i, \tilde\gamma)\Bigr)
-\sum_{1\leq j\leq x}\gamma_j
\right]
%%%%%%%%%%%%%%%%%%%%%%%%%%%%%%%%%%%%%%%%%%%%%%%%
&=\displaystyle
\underset{\substack{\beta\sim{\cal D}\\ \gamma\sim{\cal D}}}{\mathbb{E}}
\Biggl[
v_i\Bigl(x_i(\tilde\beta, \tilde\gamma)\Bigr)
\Biggr]
-
\underset{\beta\sim{\cal D}}{\mathbb{E}}
\left[
\sum_{1\leq j\leq x}\beta_j
\right]
\nonumber
%\label{equation:feldman-half-2}
\\
&\geq
\displaystyle
\underset{\substack{\beta\sim{\cal D}\\ \gamma\sim{\cal D}}}{\mathbb{E}}
\Biggl[
v_i\Bigl(x_i(\tilde\beta,\gamma)\Bigr)
\Biggr]
-
\underset{\beta\sim{\cal D}}{\mathbb{E}}
\left[
\sum_{1\leq j\leq x}\beta_j
\right]
\label{equation:feldman-half-2}
\end{align}

\noindent Notice that the second line of this derivation differs from the
right-hand side of the first line only in that the allocation $x_i(\cdot,\cdot)$
of $i$ is evaluated when he plays $\tilde\beta$ against $\gamma$, instead of
against $\tilde\gamma$. Since $\tilde\gamma$ is a ``truncation'' of $\gamma$,
the inequality $x_i(\tilde\beta,\tilde\gamma)\geq x_i(\tilde\beta,\gamma)$ is
justified. Finally, we take expectation of~(\ref{equation:post-feldman-upa})
over $\vec{b}_i\sim\tilde{\cal D}$, to obtain:

\begin{align}
&\displaystyle
\underset{
\substack{\vec{b}_i\sim\tilde{\cal D}\\ \vec{b}_{-i}\sim{\cal P}}
}{\mathbb{E}}\Bigl[
u_i^{v_i}(\vec{b}_i,\vec{b}_{-i})
\Bigr]
%%%%%%%%%%%%%%%%%%%%%%%%%%%%%%%%%%%%%%%%%%%%%%%%
%&
\geq\displaystyle
\underset{
\substack{\vec{b}_i\sim\tilde{\cal D}\\ \vec{b}_{-i}\sim{\cal P}}
}{\mathbb{E}}\Biggl[
v_i\Bigl(x_i(\vec{b}_i,\vec{b}_{-i})\Bigr)
-\sum_{j\leq x}b_i(j)
\Biggr]\nonumber\\
%%%%%%%%%%%%%%%%%%%%%%%%%%%%%%%%%%%%%%%%%%%%%%%%
&=\displaystyle
\underset{
\substack{\beta\sim{\cal D}\\ \gamma\sim{\cal D}}
}{\mathbb{E}}\Biggl[
v_i\Bigl(x_i(\tilde\beta,\gamma)\Bigr)
-
\sum_{j\in \bar{T}_{\beta}}\beta_j
\Biggr]
%%%%%%%%%%%%%%%%%%%%%%%%%%%%%%%%%%%%%%%%%%%%%%%%
=\displaystyle
\underset{
\substack{\beta\sim{\cal D}\\ \gamma\sim{\cal D}}
}{\mathbb{E}}
\Biggl[
v_i\Bigl(x_i(\tilde\beta,\gamma)\Bigr)
\Biggr]
-
\underset{\beta\sim{\cal D}}{\mathbb{E}}
\left[
\sum_{j\in \bar{T}_{\beta}}\beta_j
\right]
\label{equation:feldman-utility}
\end{align}

\noindent
By~(\ref{equation:feldman-utility}),~(\ref{equation:feldman-half-2}),~(\ref{equation:feldman-half-1})
and~(\ref{equation:feldman-2}) we derive:

\begin{equation}
\underset{
\substack{\vec{b}_i\sim\bar{{\cal D}}\\ \vec{b}_{-i}\sim{\cal P}}
}{\mathbb{E}}\Bigl[
u_i^{v_i}(\vec{b}_i,\vec{b}_{-i})
\Bigr]
\geq
\underset{\substack{\beta\sim{\cal D}\\ \gamma\sim{\cal D}}}{\mathbb{E}}
\Biggl[
v_i\Bigl(x_i(\tilde\beta, \gamma)\Bigr)
\Biggr]
-
\underset{\beta\sim{\cal D}}{\mathbb{E}}
\left[
\sum_{1\leq j\leq x}\beta_j
\right]
\end{equation}

\noindent The lower bounding by $v_i(x)/2$ of the first term of the right-hand
side of this expression is done as previously, for the Discriminatory Auction,
in derivations~(\ref{equation:feldman-trick-begin})
through~(\ref{equation:feldman-trick-end})\otrem{; as before, we ensure
existence of a pure bidding strategy $\vec{b}_i\sim\bar{{\cal D}}$ that
satisfies the requirements of Theorem~\ref{theorem:generic-stmt}
and~\eqref{equation:myFeldman-dpa}, with $(\lambda, \mu)=\left(\frac{1}{2},
1\right)$.}\qed

\subsection{Lower Bounds}
\label{subsection:lower-bounds}

%%%%%%%%%%%%%%%%%%%%%%%%%%%%%%%%%%%%%%%%%%%%%%%%%%%%%%%%%%%%%%%%%%%%%%%%%%
%%%%%%%%%%%%%%%%%%%%%%%%%%%%%%%%%%%%%%%%%%%%%%%%%%%%%%%%%%%%%%%%%%%%%%%%%%
%%%%%%%%%%%%%%%%%%%%%%%%%%%%%%%%%%%%%%%%%%%%%%%%%%%%%%%%%%%%%%%%%%%%%%%%%%
\begin{comment}
A lower bound of approximately $\frac{e}{e-1}$ for \upa s with submodular bidders was given in \cite{Markakis12}. Thus, our upper bound for this case is less than a factor $2$ away. For subadditive valuation functions, we prove in {\bf Appendix B} a lower bound of almost $2$

\begin{theorem}
The Price of Anarchy is at least $\frac{2k}{k+1}$ for the \upa\ with subadditive valuations (under the uniform bidding interface).
\label{theorem:upa-sa-lb}
\end{theorem}
\end{comment}
%%%%%%%%%%%%%%%%%%%%%%%%%%%%%%%%%%%%%%%%%%%%%%%%%%%%%%%%%%%%%%%%%%%%%%%%%%
%%%%%%%%%%%%%%%%%%%%%%%%%%%%%%%%%%%%%%%%%%%%%%%%%%%%%%%%%%%%%%%%%%%%%%%%%%
%%%%%%%%%%%%%%%%%%%%%%%%%%%%%%%%%%%%%%%%%%%%%%%%%%%%%%%%%%%%%%%%%%%%%%%%%%

The only lower bound known for the \dpa\ is due to Christodoulou {\em et
al.}~\cite{Christodoulou13}; the authors prove a lower bound of $1.109$ for an
auction involving two bidders and three items. However, they argue that
their construction does not appear to be further extendable, towards achievement
of a tight lower bound.
%No tight lower bound is known for the \dpa, although {\em Demand Reduction} (which is
%responsible for welfare loss in this format) has been observed
%previously~\cite{Krishna02,Ausubel02}. 
We prove here an \emph{impossibility result}, showing that for
the \dpa\ no bound better than $\frac{e}{e-1}$ on the Price of Anarchy can be
achieved via the proof template given in Theorem~\ref{theorem:generic-stmt}.
Similarly, for the \upa\ we rule out that a bound better than 2 on the Price of
Anarchy can be derived through this template.

\begin{theorem}
\label{theorem:lowerbound-dpa}
There is a lower bound of $\frac{e}{e-1}$ and $2$ on the Bayesian Price of
Anarchy for the \dpa\ and the \upa, respectively, with submodular valuation
functions that can be derived through the proof template given in
Theorem~\ref{theorem:generic-stmt}.
\end{theorem}

\begin{proof}
We first prove the lower bound for the \dpa . Fix $k \in \mathbb{N}$ and $\mu
\geq 0$ arbitrarily. We construct an instance of the \dpa\ with $2$ bidders,
submodular valuation functions $\vec{v}$ and bidding vectors $\vec{b}$, such
that for every possible deviation $\vec{b}'_i$ of bidder $i =  1, 2$ we have
%%%
\begin{equation*}
\sum_{i = 1}^2 u_i^{\vec{v_i}}(\vec{b}_i',\vec{b}_{-i}) 
\leq \mu\left(1-\frac{1}{e^{1/\mu}} + \frac{1}{k}\left(1 - \frac{1}{e}\right)\right) SW(\vec{v},\vec{x}^{\vec{v}}) - \mu \sum_{j=1}^k \beta_j(\vec{b}).
\end{equation*}
%%%
By taking $k$ to infinity, we see that for any fixed value of $\mu$, the best
Price of Anarchy that we can obtain using Theorem \ref{theorem:generic-stmt} is
$\max\{1,\mu\}/(\mu(1 - e^{-1/\mu}))$. The latter expression is minimized by
taking $\mu = 1$, and from this the claim follows. 

The construction of our instance is as follows: Let the valuation functions be
defined as $v_1(j) = j$ and $v_2(j) = 0$ for every $j \in [k]$. Then
$x_1^{\vec{v}} = k$ and $x_2^{\vec{v}}= 0$ and the optimal social welfare is
$SW(\vec{v},\vec{x}^{\vec{v}}) = k$. Define the bid vector $\vec{b}_1$ of bidder
1 to be the zero vector and the bid vector $\vec{b}_2$ of bidder 2 as
%%%
\begin{equation*}
b_2(j) = 
\begin{cases}
1 - \frac{k}{e^{1/\mu}(k-j+1)} & \text{ if } 1 \leq j \leq k\left(1 - \frac{1}{e^{1/\mu}}\right) + 1, \\
0 & \text{ if } j > k\left(1 - \frac{1}{e^{1/\mu}}\right) + 1.
\end{cases}
\end{equation*}
%%%
We assume that the tie-breaking rule of the auction always assigns a unit to
bidder $1$ when there is a tie. Then, if bidder $2$ bids $\vec{b}_2$, there is a
$j$ between $1$ and $k\left(1 - \frac{1}{e^{1/\mu}}\right)+1$ such that bidder
$1$ maximizes his utility when he sets all his bids equal to $b_2(j)$. Let
$\vec{b}_1' = b_2(j) \vec{1}$ for some $j$ in this range. We have
%%%
\begin{align*}
u_1^{\vec{v}_1}(\vec{b}_1', \vec{b}_{-2}) 
 & = v_1(k-j+1) - (k-j+1)b_2(j) \\
& =  \frac{k}{e^{1/\mu}} + \mu \sum_{\ell=1}^k b_2(j) - \mu \sum_{\ell=1}^k \beta_j(\vec{b}) \\
& \le \frac{k}{e^{1/\mu}} + \mu \int_1^{k(1 - e^{-1/\mu}) + 1} \left(1 - \frac{k}{e^{1/\mu}}(k-t+1)\right)dt + \\
& \quad \mu \sum_{\ell=1}^{\lceil k(1 - \frac{1}{e^{1/\mu}}) \rceil} (b_2(\ell) - b_2(\ell+1)) - \mu \sum_{\ell=1}^k \beta_j(\vec{b}) \\
& = k\mu\left(1 - \frac{1}{e^{1/\mu}}\right) +  b_2(1) - \mu \sum_{\ell=1}^k \beta_j(\vec{b}) \\
& = \mu\left(1 - \frac{1}{e^{1/\mu}} + \frac{1}{k} \left(1 - \frac{1}{e}\right)\right)SW(\vec{v},\vec{x}^{\vec{v}}) - \mu \sum_{\ell=1}^k  \beta_j(\vec{b}).
\end{align*}
%%%
For bidder 2, $u_2^{\vec{v}_2}(\vec{b}'_2, \vec{b}_{-1}) \le 0$ for every bid
vector $\vec{b}'_2$. This establishes our claim for the \dpa.  \end{proof}

\iffalse
\begin{theorem}\label{theorem:lowerbound-dpa}
Fix $k \in \mathbb{N}$ arbitrarily. There exists an instance of the \dpa\ with
$2$ bidders, submodular valuation functions $\vec{v}$ and bidding vectors
$\vec{b}$, such that for every possible deviation $\vec{b}'_i$ of bidder $i =
1, 2$ we have
%%%
\begin{equation*}
\sum_{i = 1}^2 u_i^{\vec{v_i}}(\vec{b}_i',\vec{b}_{-i}) 
\leq \left(
1-\frac{1}{e} + \frac{1}{k}\left(
1 - \frac{1}{e}\right)\right)
SW(\vec{v},\vec{x}^{\vec{v}}) - \sum_{j=1}^k \beta_j(\vec{b}).
\end{equation*}
\end{theorem}
\fi

Theorem~\ref{theorem:lowerbound-dpa} rules out the possibility of obtaining
better bounds by means of the smoothness framework of~\cite{ST13}, or by means
of {\em any} approach aiming at identifying the $\vec{b}_i'$ required by
Theorem~\ref{theorem:generic-stmt}, including~\cite{Feldman12}. These are
essentially the only known techniques for obtaining upper bounds on the Bayesian
Price of Anarchy. Thus, any improvement on our upper bound for the \dpa\ must
use either specific properties of the (Bayes-Nash equilibrium) distribution
$\mathcal{D}$, or a completely new approach altogether. 
%The same holds for improvements of the upper bound for the \upa\ below $2$ --
%and towards the only known lower bound of $\frac{e}{e-1}$
%from~\cite{Markakis12} (should it be worst-case).

\iffalse
Similarly, for Uniform Price Auctions no bound strictly better than $2$ can be
proven via Theorem~\ref{theorem:generic-stmt}. 

\begin{theorem}\label{thm:lb-smoothness-dpa}
Let $k = 1$. There exists an instance of the \upa\ with $2$ bidders, submodular
valuation functions $\vec{v}$ and bidding vectors $\vec{b}$ such that for every
possible deviation $\vec{b}'_i$ of bidder $i = 1,2$ we have
%%%
\begin{equation}\label{eq:lb-dpa-smooth}
\sum_{i = 1}^2 u_i^{\vec{v_i}}(\vec{b}_i',\vec{b}_{-i}) < \lambda SW(\vec{v},\vec{x}^{\vec{v}}) - \mu \beta_1(\vec{b}).
\end{equation}
for $\lambda > 1 - x + \mu x$ for every $x \in (0,1)$ and every $\mu \ge 0$. 
\end{theorem}

Recall that Theorem~\ref{theorem:generic-stmt} gives a bound of
$(\mu+1)/\lambda$ on the Price of Anarchy for Uniform Price Auctions. By
choosing $x = \frac12$, \eqref{eq:lb-dpa-smooth} holds for every $\lambda >
\frac12 (1+\mu)$ and $\mu \ge 0$. Thus, all bounds on the Price of Anarchy
derived through Theorem~\ref{theorem:generic-stmt} must satisfy $\lambda \le
\frac12 (1+\mu)$, which gives a lower bound of $2$. 

\fi

%Finally, in {\bf Appendix C} we also provide an example, using a discretized
%strategy space, demonstrating that Bayes-Nash equilibria are not always
%efficient in the \dpa.

%% file: smoothness.tex
\section{Smoothness and its Implications}
\label{sec:z-smoothness}

In this section we discuss implications of our developed bounds, on {\em
simultaneous and sequential compositions} of auction mechanisms, in the spirit
that they have been studied in previous
works~\cite{Christodoulou08,Bhawalkar11,Hassidim11,Syrgkanis12,ST13}. In
particular, we elaborate on the connections of our results to the {\em
smoothness framework}, recently developed by Syrgkanis and Tardos~\cite{ST13},
for bounding the inefficiency of such {\em composite} auction mechanisms. A {\em
simultaneous composition} of auction mechanisms is a mechanism ${\cal M}$
itself, wherein each of $m$ distinct goods is allocated in a separate auction,
independently of -- and simultaneously to all other goods. This setting was
first investigated by Christodoulou, Kovacs and Schapira
in~\cite{Christodoulou08}, where the authors studied the inefficiency of
simultaneous Second-Price Auctions. In the setting of~\cite{Christodoulou08},
bidders have combinatorial valuation functions, defined over subsets of goods;
however, they must bid separately (and simultaneously) for each distinct good,
in a separate Second-Price auction. In the same spirit, a {\em sequential
composition} of auction mechanisms is a mechanism ${\cal M}$, wherein each of
$m$ distinct goods is allocated in a separate auction and these auctions are
carried out sequentially. The inefficiency of this setting was first
investigated by Paes Leme, Syrgkanis and Tardos~\cite{Syrgkanis12a} and,
subsequently, by Syrgkanis and Tardos~\cite{Syrgkanis12,ST13} in the incomplete
information setting. In the context of sequential auctions, the inefficiency
upper bounds apply for notions of equilibria under incomplete information,
including {\em Perfect Bayesian} and {\em Sequential Equilibria} -- we refer the
reader to~\cite{Syrgkanis12} for a detailed discussion.

The results we describe here, concern simultaneous and sequential compositions
of {\em multi-unit} mechanisms, i.e., Discriminatory and Uniform Price Auctions;
the emerging composite mechanism enables the handling of bidders with {\em
multi-unit combinatorial} preferences over {\em multisets} of items. In the
simultaneous composition setting, every distinct auction is used for auctioning
multiple units of a single good in limited supply. In the sequential composition
setting, all available units of each distinct good are auctioned in a single
multi-unit auction; a sequence of such auctions, one per good, gives rise to the
sequential mechanism. Using the smoothness analysis framework of~\cite{ST13}, we
show how our improved bounds from the previous section, for bidders with
symmetric submodular and subadditive valuation functions, yield improved bounds
for composite mechanisms, involving bidders with valuation functions over
multisets of items.

\subsection{Valuation Functions over Multisets}
\label{subsection:multiset-vfs}

Before reviewing briefly the smoothness analysis framework, let us describe
first the classes of valuation functions (as given in~\cite{ST13}), to which the
analysis of simultaneous/sequential compositions can be applied. Consider a set
$[m]$ of $m$ distinct goods, each good $j\in[m]$ being available in limited
supply of $k_j\geq 1$ identical units. An allocation for each bidder $i$ is a
multiset of items, denoted by $\vec{x}_i=(x_i(1),\dots,x_i(m))$, where
$x_i(j)\in\{0,\dots, k_j\}$ is the number of units from good $j\in[m]$. Let
${\cal U}=\Bigl[\times_{j\in[m]}(\{0\}\cup[k_j])\Bigr]$ denote the set of all
allocation vectors $\vec{x}$ corresponding to distinct multisets of items. We
consider  valuation functions over multisets, $v:{\cal U}\mapsto\mathbb{R}^+$,
with $v(\vec{0})=0$.

\bigskip

\noindent{\bf Simultaneous Composition.} In the simultaneous auctions setting,
the smoothness framework of~\cite{ST13} allows usage of our developments from
the previous section, if the valuation functions over multisets satisfy certain
conditions ({\bf Theorem 5.1} of~\cite{ST13}), in particular if: 

\medskip

\noindent{\em The restriction $v^j$ of the valuation function $v$, to each good
$j\in[m]$, belongs to a class $V^j$ and there exist $\rho\geq 1$ sets of
valuation functions $\Bigl\{v^{j,\ell}\in V^j\,\Bigl|\Bigr.\,j\in[m]\Bigr\}$,
for $\ell=1,\dots,\rho$, such that
$v(\vec{x})=\max_{\ell}\sum_jv^{j,\ell}(x(j))$.} 

\medskip

This condition generalizes the definition of the combinatorial class of {\em
fractionally subadditive}~\cite{Feige09} valuation functions, to the premise of
multi-unit availability of each good. The class of fractionally subadditive
valuation functions contains functions $v:2^{[m]}\mapsto\mathbb{R}^+$ and is
identical to the class of \xos\ functions, identified in~\cite{Lehmann06}. The
class \xos\ is a strict superset of submodular functions over $2^{[m]}$ and a
strict subset of subadditive ones over $2^{[m]}$; a function $v$ belongs to
\xos\ (equivalently, is fractionally subadditive) if there exist $\rho$ {\em
additive} functions $\alpha^{\ell}:2^{[m]}\mapsto\mathbb{R}^+$,
$\ell=1,\dots,\rho$, with $\alpha^{\ell}(X)=\sum_{j\in X}\alpha^{\ell}(\{j\})$
for any $X\in 2^{[m]}$, so that $v(X)=\max_{\ell}\alpha^{\ell}(X)$.

A class of functions that satisfies the aforementioned condition trivially, is
the additive composition $v(\vec{x})=\sum_jv^j(x(j))$ of any set
$\Bigl\{v^j\,\Bigl|\Bigr.\,j\in[m]\Bigl\}$ of symmetric submodular multi-unit
functions $v^j$ per good. As in~\cite{ST13}, the condition is satisfied by the
more general class of submodular valuation functions over multisets, defined as
follows~\cite{Kapralov12,Krysta13}:

\begin{definition}
\label{multiset-submod}
For any $\ell=1,\dots,m$ let $\vec{e}_{\ell}$ be the unary vector with
$e_{\ell}(\ell)=1$ and $e_{\ell}(j)=0$, for $j\neq\ell$. Let $\vec{x}$ and
$\vec{y}$ denote two multisets from ${\cal U}$, so that $\vec{x}\leq\vec{y}$,
where ``$\leq$'' holds component-wise. Then, a non-decreasing function $v:{\cal
U}\mapsto\mathbb{R^+}$ is submodular if
$v(\vec{x}+\vec{e}_{\ell})-v(\vec{x})\geq
v(\vec{y}+\vec{e}_{\ell})-v(\vec{y})$.
\label{definition:multiset-sm}
\end{definition}

The inefficiency upper bounds that we present for this class of functions
improve upon the ones given in~\cite{ST13}. \vorem{Moreover, we can also derive
upper bounds for valuation functions that are ``{\em fractionally subadditive
compositions}'' of symmetric subadditive functions, in the following sense:}

\vorem{
\begin{definition}
\label{xos-subadd}
A valuation function $v:{\cal U}\mapsto\mathbb{R}^+$ is a fractionally
subadditive composition of symmetric subadditive functions, if there exist
$\rho\geq 1$ sets of symmetric subadditive functions
$\Bigl\{v^{j,\ell}:[k_j]\mapsto\mathbb{R}^+\,\Bigl|\Bigr.\,j\in[m]\Bigr\}$, for
$\ell=1,\dots,\rho$, so that $v(\vec{x})=\max_{\ell}\sum_jv^{j,\ell}(x(j))$.
\label{definition:multiset-sa}
\end{definition}
}

\vorem{
Let us note here that the class described in
Definition~\ref{definition:multiset-sa}, is strictly contained in the class of
subadditive valuation functions over multisets, which consists of the functions
satisfying $v(\vec{x})\leq v(\vec{y})+v(\vec{z})$, for
$\vec{x}=\vec{y}+\vec{z}$. The results described below ({\bf
Table~\ref{table:upper-bounds2}}) are valid for the classes of functions
described in Definitions~\ref{definition:multiset-sm}
and~\ref{definition:multiset-sa}.
}

\bigskip

\noindent{\bf Sequential Composition.} For sequential compositions, the
smoothness framework of~\cite{ST13} can yield interesting inefficiency bounds
for valuation functions $v:\times_j(\{0\}\cup[k_j])\mapsto\mathbb{R}^+$ that can
be expressed as $v(\vec{x})=\max_jv^j(x(j))$, for an appropriate selection of a
set of $m$ symmetric multi-unit valuation functions
$\Bigl\{v^j\,\Bigl|\Bigr.\,j\in[m]\Bigr\}$. Notice that this definition is
reminiscent of the class of {\em Unit-Demand} valuation functions, for the
combinatorial setting with a single unit per good (i.e., $k_j=1$). Indeed, it
was shown very recently in~\cite{Feldman13} that, even for {\em
Gross-Substitute} valuation functions (a class that strictly contains
Unit-Demand and is strictly contained in the submodular class~\cite{Lehmann06}),
sequential first-price auctions become prohibitively inefficient.

The results that we present for sequential compositions ({\bf
Table~\ref{table:upper-bounds2}}) are valid for such {\em ``Unit-Demand
Compositions''} of symmetric submodular and subadditive valuation functions.

\subsection{Smoothness Upper Bounds}

We first review the smoothness definitions introduced in \cite{ST13} (adapted to
our multi-unit auction setting). As introduced earlier, let $P_i(\vec{b})$ refer
to the payment of bidder $i$ under bidding profile $\vec{b}$. 

\begin{definition}[\cite{ST13}]\label{def:smooth}
A mechanism ${\cal M}$ is \emph{$(\lambda, \mu)$-smooth} for $\lambda > 0$ and
$\mu \ge 0$ if for any valuation profile $\vec{v}$ and for any bidding profile
$\vec{b}$ there exists a randomized bidding strategy $B_i'\equiv B_i'(\vec{v},
\vec{b}_i)$ for each bidder $i$, such that 
%%%
$$
\sum_{i \in [n]} 
\mathbb{E}_{\vec{b}_i'\sim B_i'}\Bigl[
u_i^{\vec{v}_i}(\vec{b}'_i, \vec{b}_{-i})
\Bigr]\ge 
\lambda SW(\vec{v}, \vec{x}^{\vec{v}}) - \mu \sum_{i \in [n]} P_i(\vec{b}).
$$
%%%
\end{definition}

In~\cite{ST13} it is shown that if a mechanism is $(\lambda, \mu)$-smooth, then
several results follow automatically. One such result concerns upper bounds on
the Price of Anarchy. Another result is that the smoothness property is retained
under \emph{simultaneous} and 
%\otrem{(the normal form representation of)} 
\emph{sequential compositions}. 
%In such compositions there are $m$ mechanisms with separate allocation and
%payment rules. Every bidder specifies for each mechanism a bidding profile. In
%the simultaneous composition, these profiles are submitted simultaneously,
%while in the sequential composition, they are submitted sequentially. each
%bidder expresses his valuation for the $m$-tuples of outcomes of the mechanisms
%in a restricted way. More precisely, in the simultaneous composition it is
%assumed that the valuation function of each bidder is \emph{fractionally
%subadditive} across the $m$ mechanisms (see \cite{ST13} for formal
%definitions). In the sequential composition, the valuation function of each
%bidder is defined as the maximum of his valuations over these mechanisms. 
We summarize the main composition results of \cite{ST13} in the theorem below. 

\begin{theorem}
[Theorems 4.2, 4.3, 5.1, and 5.2 in \cite{ST13}]\label{thm:st-ifm1} \mbox{}
\begin{compactenum}[(i)]
%%%
\item If ${\cal M}$ is $(\lambda, \mu)$-smooth, then the correlated (or mixed
Bayesian) Price of Anarchy of ${\cal M}$ is at most $\max\{1, \mu\}/\lambda$.
%%%
\item If ${\cal M}$ is a simultaneous (res\-pec\-ti\-ve\-ly, sequential)
composition of $m$ $(\lambda, \mu)$-smooth mechanisms, then ${\cal M}$ is
$(\lambda, \mu)$-smooth (resp., $(\lambda, \mu+1)$-smooth).
%%%
\end{compactenum}
\end{theorem}

By exploiting our Key Lemma, we can show that the \dpa\ is smooth.
Theorem~\ref{thm:smooth1} in combination with Theorem~\ref{thm:st-ifm1} leads to
the composition results stated in Table~\ref{table:upper-bounds2} (these bounds
are achieved for $\alpha = 1$).

\begin{theorem}\label{thm:smooth1}
The Discriminatory Auction is $(\lambda, \mu)$-smooth (both in the standard and
uniform bidding format) with
%%%
\begin{compactenum}[(i)]
%%%
\item $(\lambda, \mu) = (\alpha \left(1 - e^{-1/\alpha}\right), \alpha)$ for
submodular valuation functions, and 
%%%
\vorem{
\item $(\lambda, \mu) = (\frac{\alpha}{2} \left(1 - e^{-1/\alpha}\right),
\alpha)$ for subadditive valuation functions.
}
\end{compactenum}
\end{theorem}
In particular, we obtain as a corollary, by using Theorem \ref{thm:st-ifm1},
improved upper bounds for the classes of valuation functions defined in
Definitions \ref{multiset-submod} \vorem{and \ref{xos-subadd}.}
 
\begin{proof}
Note that for the \dpa\ we have 
%%%
$$
\sum_{i \in [n]} P_i(\vec{b})
=
\sum_{i \in [n]} \sum_{j = 1}^{x_i(\vec{b})} b_i(j) 
= 
\sum_{j = 1}^k \beta_j(\vec{b}).
$$
%%%
Note that Lemma~\ref{lem:smoothness-key} holds for every bidder $i \in [n]$ with
$\lambda$ and $\mu$ as stated in Theorem~\ref{thm:smooth1} (see also the proof
of Theorem~\ref{thm:main-sm}). By invoking Lemma~\ref{lem:smoothness-key} and
summing inequality \eqref{eq:rand-smooth} over all bidders, we obtain
%%%
\begin{align*}
\sum_{i \in [n]}\mathbb{E}_{\vec{b}_i'\sim B_i'}\Bigl[
u_i^{\vec{v}_i}(\vec{b}'_i, \vec{b}_{-i})\Bigr] 
& \ge 
\lambda \sum_{i \in [n]} v_i(x_i^\vec{v}) - 
\mu \sum_{i \in [n]} \sum_{j = 1}^{x^\vec{v}_i}\beta_j(\vec{b}_{-i}) 
 \ge 
\lambda SW(\vec{v}, \vec{x}^{\vec{v}}) - 
\mu \sum_{j = 1}^k \beta_j(\vec{b}), 
\end{align*}
%%%
where the last inequality holds because for every bidder $i$,
$\beta_j(\vec{b}_{-i}) \le \beta_j(\vec{b})$ for every $j = 1, \dots, k$, and
$\sum_{i \in [n]} x_i(\vec{b}) = k$ and $\beta_j(\vec{b})$ is non-decreasing in
$j$.
\end{proof}

\bigskip

\begin{table}[t]
\begin{center}
\begin{tabular}{c@{\quad}|@{\quad}c@{\quad}c@{\qquad}c@{\quad}c}
\multirow{3}{*}{{\bf Valuation Functions}$^*$} & \multicolumn{2}{l}{{\bf Discriminatory Auction}} & \multicolumn{2}{c}{{\bf \upa}} \\
 & {\em Simultaneous } & {\em Sequential } & {\em Simultaneous } & {\em Sequential } \\
\hline
\multirow{3}{*}{{\em Submodular}} & & & & \\
 & $\displaystyle\frac{e}{e-1}$ & $\displaystyle\frac{2e}{e-1}$ & \multicolumn{2}{c}{$3.1462$}\\
 & & & & \\
\multirow{3}{*}{\vorem{{\em Subadditive}}} & & & & \\
 & \vorem{$\displaystyle\frac{2e}{e-1}$} & \vorem{$\displaystyle\frac{4e}{e-1}$} & \multicolumn{2}{c}{\vorem{$6.2924$}}\\
 & & & &
\end{tabular}
\end{center}
\caption{
Upper bounds on the Bayesian Price of Anarchy for compositions of Discriminatory
and Uniform Price Auctions. \vorem{$^*$The valuation function classes mentioned
here correspond (by abuse of their names) to the ones described in
subsection~\ref{subsection:multiset-vfs}}.
}
\label{table:upper-bounds2}
\end{table}

\iffalse
Note that the bids in the support of the random deviation $\vec{b}'_i$ with
$\alpha = 1$ in the proof of Lemma~\ref{lem:smoothness-discriminatory} satisfies
%%%
$$
\max_{\vec{b}_{-i}} P_i(\vec{b}'_i, \vec{b}_{-i}) = \left(1-\frac{1}{e}\right) \frac{v_i(x^{\vec{v}}_i)}{x^{\vec{v}}_i}\le \max_{x \in [k]} v_i(x).
$$
%%%
where $P_i(\vec{b}) = \sum_{j = 1}^{x_i(\vec{b})} b_i(j)$ refers to the payment
of bidder $i$ under bidding profile $\vec{b}$.
Lemma~\ref{lem:smoothness-discriminatory} above therefore satisfies the stronger
definition of \emph{conservative smoothness} according to Definition~8.1 in
\cite{ST13} which allows us to use Theorem 8.2 in \cite{ST13}: 

\begin{corollary}
The correlated price of anarchy and the Bayesian price of anarchy of
discriminatory multi-unit auction with submodular valuation functions and where
each bidder has a budget constraint is at most $\frac{e}{e-1}$.
\end{corollary}
\fi

For auction mechanisms where one needs to impose a no-overbidding assumption, a
different smoothness notion is introduced in \cite{ST13}. Given a mechanism
${\cal M}$, define bidder $i$'s \emph{willingness-to-pay} as the maximum payment
he could ever pay conditional to being allocated $x$ units, i.e.,
$W_i(\vec{b}_i, x) = \max_{\vec{b}_{-i}: x_i(\vec{b}) = x} P_i(\vec{b})$.
%%%
\begin{definition}[\cite{ST13}]
A mechanism ${\cal M}$ is \emph{weakly} $(\lambda, \mu_1, \mu_2)$-smooth for
$\lambda > 0$ and $\mu_1, \mu_2 \ge 0$ if for any valuation profile $\vec{v}$
and for any bidding profile $\vec{b}$ there exists a randomized bidding profile
$B_i'\equiv B_i'(\vec{v}, \vec{b}_i)$, for each bidder $i$, such that 
%%%
$$
\sum_{i \in [n]}
\mathbb{E}_{\vec{b}_i'\sim B_i'}\Bigl[
u_i^{\vec{v}_i}(\vec{b}'_i, \vec{b}_{-i})\Bigr]
\ge 
\lambda SW(\vec{v}, \vec{x}^{\vec{v}}) - \mu_1 \sum_{i \in [n]} P_i(\vec{b}) - \mu_2 \sum_{i \in [n]} W_i(\vec{b}_i, x_i(\vec{b})).
$$
%%%
\end{definition}

\noindent Syrgkanis and Tardos \cite{ST13} establish the following results. 
\iffalse
\begin{theorem}[Theorem~7.4 in \cite{ST13}]\label{thm:st-ifm3}
Let ${\cal M}$ be a mechanism that is weakly $(\lambda, \mu_1, \mu_2)$-smooth
and satisfies the no-overbidding assumption. Then the correlated (resp., mixed
Bayesian) Price of Anarchy is at most $(\mu_2 + \max\{1, \mu_1\})/\lambda$.
\end{theorem}
\fi
\begin{theorem}[Theorems~7.4, C.4 and C.5 in \cite{ST13}]\mbox{}
%%%
\begin{compactenum}[(i)]
%%%
\item If ${\cal M}$ is $(\lambda, \mu_1, \mu_2)$-weakly smooth, then the
correlated (or mixed Bayesian) Price of Anarchy of ${\cal M}$ is at most $(\mu_2
+ \max\{1, \mu_1\})/\lambda$.
%%%
\item If ${\cal M}$ is a simultaneous (resp., sequential) composition of $m$
$(\lambda, \mu_1, \mu_2)$-weakly smooth mechanisms, then ${\cal M}$ is
$(\lambda, \mu_1, \mu_2)$-weakly smooth (resp., $(\lambda, \mu_1+1,
\mu_2)$-weakly smooth).
%%%
\end{compactenum}
\end{theorem}

Using our Key Lemma, we can show that the Uniform Price Auction is weakly
smooth. As a consequence, we obtain the composition results stated in
Table~\ref{table:upper-bounds2} (these bounds are achieved for $\alpha =
-1/({\cal W}_{-1}(-1/e^2) + 2) \approx 0.87$).

\begin{theorem}
\label{thm:smooth2}
The Uniform Price Auction is weakly $(\lambda, \mu_1, \mu_2)$-smooth (both in
the standard and uniform bidding format) with
%%%
\begin{compactenum}[(i)]
%%%
\item $(\lambda, \mu_1, \mu_2) = (\alpha \left(1 - e^{-1/\alpha}\right), 0,
\alpha)$ for submodular valuation functions, and 
%%%
\vorem{
\item $(\lambda, \mu_1, \mu_2) = (\frac{\alpha}{2} \left(1 -
e^{-1/\alpha}\right), 0, \alpha)$ for subadditive valuation functions.
}
%%%
\end{compactenum}
%%%
\end{theorem}

\begin{proof}
The following transformation will be helpful in the poof of
Theorem~\ref{thm:smooth2}. Let $\vec{b}$ be an arbitrary bidding profile. We
derive a uniform bidding profile $\bar{\vec{b}}$ from $\vec{b}$ as follows: Let
$c_i = b_i(x_i(\vec{b}))$ be the last winning bid of bidder $i \in [n]$; for
ease of notation, we adopt the convention that $b_i(0) = 0$. Define
$\bar{\vec{b}}_i$ as the vector that is $c_i$ on the first $x_i(\vec{b})$
entries and zero everywhere else. Clearly, $\bar{\vec{b}}$ is a uniform bidding
profile. 

\begin{lemma}\label{lem:prop}
Let $\bar{\vec{b}}$ be the uniform bidding profile derived from $\vec{b}$. Then
for the Uniform Price Auction, the following holds for every bidder $i \in [n]$:
%%%
\begin{enumerate}[(i)]
%%%
\item $x_i(\bar{\vec{b}}) = x_i(\vec{b})$;
%%%
\item $W_i(\bar{\vec{b}}_i, x_i(\bar{\vec{b}})) = W_i({\vec{b}}_i, x_i({\vec{b}}))$.
%%%
\end{enumerate}
%%%
\end{lemma}

We first prove weak smoothness for the uniform bidding format and then extend
this result to the standard bidding format via a coupling argument. 

It is not hard to verify that for the Uniform Price Auction we have
$W_i(\vec{b}_i, x) = x b_i(x)$. As before, exploiting
Lemma~\ref{lem:smoothness-key} and summing inequality \eqref{eq:rand-smooth}
over all bidders we obtain
%%%
\begin{align*}
\sum_{i \in [n]} \mathbb{E}_{\vec{b}_i'\sim B_i'}\Bigl[
u_i^{\vec{v}_i}(\vec{b}'_i, \vec{b}_{-i})\Bigr] 
& \ge 
\lambda SW(\vec{v}, \vec{x}^{\vec{v}}) - 
\mu \sum_{j = 1}^k \beta_j(\vec{b}). 
\end{align*}
%%%
If $\vec{b}$ is a uniform bidding profile then the claim follows because 
%%%
\begin{equation}\label{eq:rel-B}
\sum_{j = 1}^k \beta_j(\vec{b}) 
= \sum_{i \in [n]} x_i(\vec{b}) b_i(x_i(\vec{b}))
= \sum_{i \in [n]} W_i(\vec{b}_i, x_i(\vec{b})).
\end{equation}

Note that for the standard bidding format, the first equality would be false
because we can only infer that $\sum_{j} \beta_j(\vec{b}) \ge \sum_{i}
x_i(\vec{b}) b_i(x_i(\vec{b}))$. However, the following work-around establishes
weak smoothness for the Uniform Price Auction and the standard bidding format. 

Note that in general $P_i(\bar{\vec{b}}) \neq P_i(\vec{b})$ because all losing
bids in $\bar{\vec{b}}$ are zero. However, the above two properties turn out to
be sufficient to prove weak smoothness in the standard bidding format: Let
$\bar{\vec{b}}$ be the uniform bidding profile that we obtain from $\vec{b}$ as
described above. Applying Lemma~\ref{lem:smoothness-key} to the uniform bidding
profile $\bar{\vec{b}}$ and pricing rule $P_i(\bar{\vec{b}}) =
W_i(\bar{\vec{b}}, x_i(\bar{\vec{b}}))$ (which is discriminatory price
dominated), we conclude that for every bidder $i \in [n]$ there exists a random
uniform bidding strategy $B_i'$, $\vec{b}'_i$ such that 
%%%
\begin{equation}\label{eq:start}
\mathbb{E}_{\vec{b}_i'\sim B_i'}\Bigl[
v_i(x_i(\vec{b}'_i, \bar{\vec{b}}_{-i})) - W_i(\vec{b}'_i, x_i(\vec{b}'_i, \bar{\vec{b}}_{-i}))\Bigr] 
\ge 
\lambda v_i(x_i^{\vec{v}}) - 
\mu \sum_{j = 1}^{x^\vec{v}_i}\beta_j(\bar{\vec{b}}_{-i}).
\end{equation}
%%%
Note that by Lemma~\ref{lem:prop}
%%%
\begin{align}
u_i^{\vec{v}_i}(\vec{b}'_i, \vec{b}_{-i}) 
& = v_i(x_i(\vec{b}'_i, \vec{b}_{-i})) - P_i(\vec{b}'_i, \vec{b}_{-i}) \\
& \ge v_i(x_i(\vec{b}'_i, \vec{b}_{-i})) - W_i(\vec{b}'_i, x_i(\vec{b}'_i, \vec{b}_{-i})) \notag \\
&  = v_i(x_i(\vec{b}'_i, \bar{\vec{b}}_{-i})) - W_i(\vec{b}'_i, x_i(\vec{b}'_i, \bar{\vec{b}}_{-i})). \label{eq:coupling1}
\end{align}
%%%
By summing inequality \eqref{eq:start} over all bidders and exploiting
\eqref{eq:coupling1}, we thus obtain 
%%%
\begin{align*}
\sum_{i \in [n]}
\mathbb{E}_{\vec{b}_i'\sim B_i'}\Bigl[
u_i^{\vec{v}_i}(\vec{b}'_i, \vec{b}_{-i})
\Bigr] 
& \ge 
\lambda
SW(\vec{v}, \vec{x}^{\vec{v}}) - 
\mu
\sum_{j = 1}^{k}\beta_j(\bar{\vec{b}})  
 = 
\lambda SW(\vec{v}, \vec{x}^{\vec{v}}) - 
\mu \sum_{i \in [n]} W_i(\bar{\vec{b}}_i, x_i(\bar{\vec{b}})) \\
& = 
\lambda SW(\vec{v}, \vec{x}^{\vec{v}}) - 
\mu \sum_{i \in [n]} W_i(\vec{b}_i, x_i(\vec{b}))
\end{align*}
%%%
where the first equality follows from \eqref{eq:rel-B} because $\bar{\vec{b}}$
is a uniform bidding profile and the second equality holds because of
Lemma~\ref{lem:prop}.
\end{proof}

\bigskip

\begin{proofof}{Lemma~\ref{lem:prop}}
Let $\beta_1$ and $\beta_0$ refer to the smallest winning bid and largest losing
bid under bidding profile $\vec{b}$, respectively. Note that every winning bid
under $\bar{\vec{b}}$ is at least $\beta_1$ and can only increase under
$\vec{b}$. Also, every losing bid under $\bar{\vec{b}}$ is at most $\beta_0$
under $\vec{b}$. We conclude that a bid is winning under $\bar{\vec{b}}$ if and
only if it is winning under $\vec{b}$, which proves (i). 

Recall that $W_i(\bar{\vec{b}}_i, x_i(\bar{\vec{b}})) = x_i(\bar{\vec{b}})
\bar{b}_i(x_i(\bar{\vec{b}}))$. Now, using (i) and the definition of
$\bar{\vec{b}}_i$, we obtain: 
%%%
$
W_i(\bar{\vec{b}}_i, x_i(\bar{\vec{b}}))
= x_i(\bar{\vec{b}}) \bar{b}_i(x_i(\bar{\vec{b}}))
= x_i({\vec{b}}) \bar{b}_i(x_i({\vec{b}}))
= x_i({\vec{b}}) {b}_i(x_i({\vec{b}}))
= W_i({\vec{b}}_i, x_i({\vec{b}}))
$,
%%%
which shows (ii). 
\end{proofof}

%%%%%%%%%%%%%%%%%%%%%%%%%%%%%%%%%%%%%%%%%%%%%%%%%%%%%%%%%%%%%%%%%%%%%
%%%%%%%%%%%%%%%%%%%%%%%%%%%%%%%%%%%%%%%%%%%%%%%%%%%%%%%%%%%%%%%%%%%%%
%%%%%%%%%%%%%%%%%%%%%%%%%%%%%%%%%%%%%%%%%%%%%%%%%%%%%%%%%%%%%%%%%%%%%
\begin{comment}
\bigskip

\noindent {\bf \sout{Some additional results on mechanisms with \emph{budgets}
(see \cite{ST13}) can be inferred from Theorems~\ref{thm:smooth1} and
\ref{thm:smooth2}. We defer further details to the full version of the paper.}}
\end{comment}
%%%%%%%%%%%%%%%%%%%%%%%%%%%%%%%%%%%%%%%%%%%%%%%%%%%%%%%%%%%%%%%%%%%%
%%%%%%%%%%%%%%%%%%%%%%%%%%%%%%%%%%%%%%%%%%%%%%%%%%%%%%%%%%%%%%%%%%%%
%%%%%%%%%%%%%%%%%%%%%%%%%%%%%%%%%%%%%%%%%%%%%%%%%%%%%%%%%%%%%%%%%%%%

\iffalse
\begin{corollary}
The correlated (mixed Bayesian) Price of Anarchy of the Uniform Price Auction
with submodular valuation functions is at most $3.1462 < \frac{2e}{e-1}$.
\end{corollary}
\fi

%% file: conclusions.tex
\section{Conclusions}
\label{sec:conclusions}

We derived inefficiency upper bounds in the incomplete information model for the
widely popular Discriminatory and Uniform Price Auctions, when bidders have
submodular or subadditive valuation functions. Notably, our bounds for
subadditive valuation functions already improve upon the ones that were known
for submodular bidders~\cite{Markakis12,ST13}. Moreover, for each of the two
formats and valuation function classes we considered both the {\em standard}
bidding interface~\cite{Krishna02,Milgrom04} and a practically motivated {\em
uniform} bidding interface.
To derive our results, we elaborated on several techniques from the recent
literature on {\em Simultaneous
Auctions}~\cite{ST13,Feldman12,Christodoulou08,Bhawalkar11}. By the recent
developments of~\cite{ST13}, our bounds for submodular bidders yield improved
inefficiency bounds for {\em simultaneous} and {\em sequential}  compositions of
the considered formats. In absence of an indicative lower bound in the
incomplete information model, we showed that our upper bound of $\frac{e}{e-1}$
for the Discriminatory Auction with submodular valuation functions is best
possible, w.r.t. the currently known proof techniques. Further reducing our
upper bound for the \upa\ (with submodular bidders), poses a particularly
challenging problem, given the lower bound of $2$. For the case of subadditive
bidders, the open problems involve reducing further the upper bounds (for both
auctions) under the uniform bidding format, when bidders have subadditive
valuation functions. For the Uniform Price Auction in particular, improvements
should be attainable also under the standard bidding interface.